\documentclass[10pt, journal]{IEEEtran}
\IEEEoverridecommandlockouts

\usepackage{epsfig, amsthm,amsmath,amssymb,epsf,cite,scalefnt,epstopdf,setspace, times}
\usepackage{bbm}
\usepackage{multirow}
\usepackage{algorithm}
\usepackage{algpseudocode}
\usepackage{graphicx}
\usepackage{subfig}
\usepackage{color}
\usepackage{url}
\usepackage{multicol}
\usepackage{times}
\usepackage{hyperref}

\definecolor{mygray}{gray}{0.6}

\def\tcg {\textcolor{black}}

\def\B{\mathcal{B}}
\def\cA{{\mathcal{A}}} \def\cB{{\mathcal{B}}}  
  \def\cG{{\mathcal{G}}} 
   
   \def\cP{{\mathcal{P}}}
  \def\cS{{\mathcal{S}}} 
\def\cU{{\mathcal{U}}}

   \def\bx{{\mathbf{x}}} \def\by{{\mathbf{y}}}

   \def\bX{{\mathbf{X}}}

\def\span{\mathop{\mathrm{span}}}

\renewcommand{\baselinestretch}{1.0}

     \def\d4{\!\!\!\!}

 \def\bPhi{\mathbf{\Phi}}    
 
\def\Lam{\Lambda}

\def\B{{\mathbb{B}}}
       \def\I{{\mathbb{I}}}

  \def\-{\! - \!}  \def\+{\! + \!}  \def\={\! = \!}  \def\>{\! > \!} \def\nn{\nonumber}

\def \log{\mathrm{log}}
\def\exp{\mathrm{exp}}

\newtheorem{lemma}{Lemma}

\newtheorem{remark}{Remark}

\newcommand{\bef}{\begin{figure}}
\newcommand{\eef}{\end{figure}}
\newcommand{\beq}{\begin{eqnarray}}
\newcommand{\eeq}{\end{eqnarray}}

\title{Link Quality-Guaranteed Minimum-Cost Millimeter-Wave Base Station Deployment }
\author{
Miaomiao Dong, Taejoon Kim, Minsung Cho, Kangeun Lee, and Sungrok Yoon  \thanks{M. Dong is with the Department of Electrical Engineering, City University of Hong Kong and the Department of Electrical Engineering and Computer Science, University of Kansas, Lawrence, KS, USA, emails: miao4600@163.com. T. Kim is with the Department of Electrical Engineering and Computer Science, University of Kansas, Lawrence, KS, USA, email: taejoonkim@ku.edu. M. Cho, K. Lee, and S. Yoon are with Samsung Electronics, Suwon, Korea, emails: \{ke.lee, msstar.cho, sr.eric.yoon\}@samsung.com.

This work was supported in part by Samsung.
The work of Taejoon Kim was supported in part by the National Science Foundation (NSF) under Grants CNS1955561 and AST2037864.}
}
\begin{document}
\maketitle

%\tcg{Newly added context.  }
%\tcgr{Deleted context.  }
%\tcr{Comments.  }

\begin{abstract}
Today's growth in the volume of wireless devices coupled with the promise of supporting data-intensive 5G-\&-beyond use cases is driving the industry to deploy more millimeter-wave (mmWave) base stations (BSs). Although mmWave cellular systems can carry a larger volume of traffic, dense deployment, in turn, increases the BS installation and maintenance cost, which has been largely ignored in their utilization.
In this paper, we present an approach to the problem of mmWave BS deployment in urban environments by minimizing BS deployment cost subject to BS association and user equipment (UE) outage constraints.
\tcg{By exploiting the macro diversity, which enables each UE to be associated with multiple BSs, we derive an expression for UE outage that integrates} physical blockage, UE access-limited blockage, and signal-to-interference-plus-noise-ratio (SINR) outage into its expression.
The minimum-cost BS deployment problem is then formulated as integer non-linear programming (INP).
The combinatorial nature of the problem motivates the pursuit of the optimal solution by decomposing the original problem into the two separable subproblems, i.e., cell coverage optimization and minimum subset selection subproblems.
We provide the optimal solution and theoretical justifications of each subproblem.
The simulation results demonstrating UE outage guarantees of the proposed method are presented.
Interestingly, the proposed method produces a unique distribution of the macro-diversity orders over the network that is distinct from other benchmarks.

\end{abstract}

%% ---------------------------------------------------------------------------------------------------
%% ---------------------------------------------------------------------------------------------------
%% ---------------------------------------------------------------------------------------------------
%% ----------------------------------- Introduction ------------------------------------------------
\section{Introduction}~\label{sec:introduction}
%% ---------------------------------------------------------------------------------------------------
%% ---------------------------------------------------------------------------------------------------
%% ---------------------------------------------------------------------------------------------------
%% ---------------------------------------------------------------------------------------------------

%Wireless operators have been shifting operating frequencies upward to the millimeter wave (mmWave) bands to accommodate more broadband spectrum. % \tcgr{($30$-$300$ GHz)} \tcg{(In the pathloss model, we use the $28$GHz bands, which is not in the range of $30$-$300$ GHz. Hence, I remove the "($30$-$300$ GHz)".)} to accommodate more broadband spectrum%.
Communications in the millimeter-wave (mmWave) bands will play a crucial role in facilitating the data-intensive fifth-generation (5G) use cases, including real-time  machine-type communications (MTC), interactive on-line learning, and augmented reality (AR) streaming.
The potential of the mmWave band has made it as one of the important aspects of future cellular networks \cite{Kims_gen, mmWave_magzine, mmWave_magzine2}.
While it is true that mmWave bands provide very high rate connectivity, contrary to general belief, this does not necessarily translate to high achievable throughput due to significant differences between systems operating in mmWave and legacy sub-$6$GHz bands.

\subsection{\tcg{MmWave Channel Access Challenges}}\label{challenges}
\tcg{The initial channel access in the mmWave cellular environment is a very critical problem, especially, using directional analog-digital  beamforming \cite{HealmmWave,Hadi2016,BSCap_limit2, Zhang18,aych_RFChain,Alkhateeb}.
The reliable user equipment (UE) access is critically limited by physical blockage and UE-access saturation in dense urban environments.
Specifically, weak diffraction and penetration due to the high pathloss in mmWave bands make channels susceptible to \emph{physical blockage} by random obstacles \cite{Link_status, RobertHealth1}.
If the channel does not experience physical blockage, the UE can attempt channel access to BS. 
However, the number of concurrently served UEs by a BS, which is equipped with hybrid analog-digital multiple-input multiple-output (MIMO) arrays, is strictly limited by the number of available radio frequency (RF) chains   \cite{HealmmWave,Hadi2016,BSCap_limit2,Zhang18,aych_RFChain,Alkhateeb}. When the number of active UEs in a cell is larger than the number of RF chains, \emph{UE access-limited blockage} occurs.}
Even without the physical and {UE access-limited} blockages, accumulated interference from surrounding mmWave small-cells can potentially lower the signal-to-interference-plus-noise ratio (SINR) of each UE, causing \emph{SINR outage} \cite{Mac_div2}.
Both physical and UE access-limited blockages, as well as the SINR outage, will lead to an unsatisfactory user experience.

Recently, the physical blockage challenge has been addressed by introducing macro diversity \cite{Mac_div1, UEAssoc_2,UEAssoc_3} that allows for each UE in an area to be covered by multiple BSs.
Provided macro diversity, if a link of a UE is blocked, the link can be restored by another BS that also covers the UE.
However, this benefit comes at the price of the growing number of deployed BSs.
%A larger number of BSs can also provide more \tcg{capability} \tcgr{capacity} for UE access \tcg{in the network}, mitigating the \tcg{saturated user access} \tcgr{capacity-limited} blockage.
The recent increase in the cost of deploying and maintaining small-cell BSs is a practical concern that wireless service providers are constantly facing \cite{whitepapaer_smallcell}.

%=============================================
\subsection{\tcg{Related Work}}
An efficient strategy for BS deployment is minimizing the number of deployed BSs subject to per-UE quality constraints.
In this category, BS density optimization \cite{Stoc_BS1, Stoc_BS1_ZhaoXidian, Stoc_BS3,Stoc_BS4,Stoc_BS5,my_Access}  and site-specific BS deployment  \cite{MinBSnum_VT,sub6_siteBS3,sub6_siteBS2,sub6_siteBS1} methods have been previously studied.
The BS density optimization methods \cite{Stoc_BS1, Stoc_BS1_ZhaoXidian, Stoc_BS3, Stoc_BS4, Stoc_BS5,my_Access}  rely on stochastic geometry to find minimum BS density subject to cell coverage  constraints.
While these prior works \cite{Stoc_BS1, Stoc_BS1_ZhaoXidian, Stoc_BS3, Stoc_BS4, Stoc_BS5,my_Access}  provide theoretical insights into the distribution of BSs, they rather fit traditional macro-cellular environments.
In the context of small-cell, the site-specific BS deployment methods \cite{MinBSnum_VT,sub6_siteBS3,sub6_siteBS2,sub6_siteBS1} have been studied, which find the minimum number of BSs installed on predetermined candidate locations.
%In the context of a small-cell network, the site-specific BS deployment~{techniques} \cite{MinBSnum_VT,sub6_siteBS3,sub6_siteBS2,sub6_siteBS1} find the minimum number of BSs from predetermined candidate locations to satisfy UE QoS requirements in
%have been studied in sub-$6$GHz bands.These methods find a minimum subset of candidate BS locations to satisfy UE QoS constraints.
%The above site-specific BS deployment techniques have been largely investigated in
%sub-$6$GHz networks.
The underlying assumption of these works was omnidirectional and penetrable wave propagation in the sub-$6$GHz bands, which cannot be extended to mmWave.
%due to the directional wave propagation and severe blockage as well as weak diffraction.

Incorporating mmWave pathloss models, site-specific BS deployment techniques have been studied to maximize line-of-sight (LoS) link distance given a fixed number of BSs \cite{Site_BS2, Site_BS3, Site_BS4, 7794888}. \tcg{While the works in \cite{Site_BS2, Site_BS3, Site_BS4, 7794888} considered general urban geometry, the work in \cite{ZhangYue} focused on a specific Manhattan-type geometry to maximize the 
macro diversity. 
%Maximizing the macro-diversity order in Manhattan-type geometry was the focus of the work in \cite{ZhangYue}.
Recently, the minimum-number BS deployment problem  has been studied for guaranteeing the average receive signal power \cite{mmWaveSiteDeploy}, link-blockage tolerance  level \cite{my_PartialSamplingBS}, and beam alignment reliability due to random UE rotation \cite{MmWaveBsDeploy_UEOre}.
Unlike the prior works  \cite{Site_BS2, Site_BS3, Site_BS4, 7794888, mmWaveSiteDeploy, MmWaveBsDeploy_UEOre}, the methodologies in \cite{my_PartialSamplingBS,ZhangYue} integrate the mmWave physical blockage and UE access-limited blockage models into their problem formulations.
However, these works either are limited to a specific Manhattan-type geometry \cite{ZhangYue} or ignore important link-quality measures \cite{my_PartialSamplingBS,ZhangYue}, such as SINR that must be taken in when designing and evaluating the performance of a BS deployment method.}
Thus, it is of great interest to develop a rather pragmatic strategy
that accounts for various mmWave-link-quality-related constraints.
Such approaches must be flexible to be applicable to any urban geometry and can incorporate both the installation cost and the number of BSs into its objective \cite{whitepapaer_smallcell}.

%\taejoon{Given the content provided in the previous paragraphs (related to macro diversity), I am doubt if the reviewers can quickly understand what previously written here. I still cannot clearly understand it. So, I removed them for the time being, which I think is difficult to be explained in Introduction.}
%poses challenges as modeling the random blockage effects and optimizing the deployment parameters
%(\tcr{When writing this sentence, do you have a clear idea in your mind about what we want to say here? I am sorry but I cannot really understand the exact meaning of this sentence. It is very difficult to decipher.  When the reviewer said ``pretentious", I guess he/she meant this style of writing. Remove?})
%\tcgr{Moreover, ...}

In this paper, we consolidate the latter missing components and formulate the BS deployment problem into
%minimum-cost BS deployment problem subject to UE outage constraint
 integer nonlinear programming (INP).
About existing work, this aspect has some similarities to the  well-investigated problems on optimizing BS sleeping and user association at sub-$6$GHz bands.
The latter problems were often addressed by formulating INP with the objectives of maximizing network throughput \cite{sub6_siteBS5,lagrangian_dual} or minimizing the power consumption \cite{sub6_siteBS3,sub6_siteBS2,sub6_siteBS1,sub6_siteBS4,BSSleep_Femto}. 
%\taejoon{Why not cite \cite{BSSleep_Femto} here?}.
\tcg{Because large-scale INP is NP-hard and is generally very difficult to be optimally solved}, devising suboptimal but efficient algorithms was the focus of these approaches, by using  greedy heuristics \cite{sub6_siteBS1, sub6_siteBS2,sub6_siteBS3}, Lagrangian dual \cite{BSSleep_Femto,lagrangian_dual}, and sequential subproblem formulations \cite{sub6_siteBS4,sub6_siteBS5}.
Though these subproblem frameworks were largely benefited from the deterministic link models at sub-$6$GHz bands, such models and associated problem formulations cannot be extended to the mmWave.
Although prior approaches \cite{sub6_siteBS1, sub6_siteBS2,sub6_siteBS3,sub6_siteBS4,sub6_siteBS5,BSSleep_Femto, lagrangian_dual} deal with INP,
they are different from our proposed approach in terms of the objective functions, associated constraints, and thereby, the developed algorithms.

%===================================================
\subsection{\tcg{Overview of Methodology and Contributions}}
\tcg{We address important mmWave connectivity challenges in a 3-dimensional (3D) urban  geometry  by proposing a link quality-guaranteed minimum-cost mmWave BS deployment technique.
Our contributions are summarized below.}
\begin{itemize}
  \item
  \tcg{We introduce a mmWave link state model that captures the \emph{randomness} of physical blockage, UE access-limited blockage, and SINR outage events.
  The model accounts for random locations of UEs and obstacles in stochastic geometry, which allows us to mathematically formulate the UE outage and BS association constraints.
  However, these expressions are not analyzable and difficult to be used for formulating optimization problems.
  We show instead that the UE outage constraint can be upper bounded to derive a tractable expression.
  Moreover, we verify that the UE access-limited blockage constraint can be equivalently transformed to an analytic (linear) model, facilitating a tractable optimization framework.}

\item \tcg{Next, we present the  minimum-cost mmWave BS deployment problem as large-scale INP.
However, the formulated INP is difficult to be solved directly.
While this could motivate us to pursue suboptimal algorithms 
(e.g., \cite{sub6_siteBS1, sub6_siteBS2,sub6_siteBS3,BSSleep_Femto,lagrangian_dual,sub6_siteBS4,sub6_siteBS5,mmWaveSiteDeploy,my_PartialSamplingBS}), we instead show that the formulated INP can be decomposed into two \emph{separable} subproblems, i.e., (i) cell coverage optimization problem and (ii) minimum subset BS selection problem, and  optimally solved by sequentially solving these two subproblems.
In particular, we show that the second subproblem (minimum subset BS selection problem), which is also INP, can be transformed into a linear equivalent form that can be efficiently solved via existing software.}

\item
\tcg{Finally, we evaluate the performance of our proposed designs via numerical results.
We show that the proposed scheme provides UE outage guarantees with the minimum-cost BS deployment, in contrast to other benchmark schemes \cite{mmWaveSiteDeploy,MmWaveBsDeploy_UEOre}.
An interesting aspect of the proposed scheme is that the optimized results present a unique distribution of the macro-diversity orders over the network that lowers the concentration of macro-diversity orders compared to other benchmark schemes in \cite{mmWaveSiteDeploy,MmWaveBsDeploy_UEOre}.
This is the underlying reason of the improved performance guarantee compared to the benchmarks, while deploying a reduced number of BSs.}

%the macro diversity-constrained and physical blockage-constrained minimum-cost BS deployment schemes.
\end{itemize}

The rest of the paper is organized as follows.
We present the system models in Section~\ref{sec_system_setup}.
The UE outage-guaranteed minimum-cost mmWave BS deployment problem is formulated in Section~\ref{sec_ProblemForm}.
The  proposed algorithm that solves the problem is enunciated in Section~\ref{sec:two_methods}.
The simulation results and conclusions are provided in Section \ref{Sec_Simulation} and \ref{Sec_con}, respectively.
\tcg{For ease of reference, TABLE~\ref{tab:notation} summarizes the
main variables which will be used throughout this paper.}

%\taejoon{Throughout the paper, you are mixing Table vs. TABLE. Keep the consistency}

%\begin{figure*}%
%	\centering
%	\subfloat[3D campus of the {University of Kansas}.]{{\includegraphics[width=6.5cm]{graph/geometry_3D.pdf} }}%
%	\qquad
%	\subfloat[Illustration of campus with {$B=245$} candidate BS locations (blue dots) on the building walls and the square grids partitioning the outdoor campus.]{{\includegraphics[width=9.7cm]{graph/geometry_2D_KU.pdf} }}%
%	\caption{Campus map of the University of Kansas for outdoor mmWave BS deployment.}%
%	\label{fig:geometry}
%\end{figure*}

%%================================================================================================
%%================================================================================================
%%================================================================================================
%---------------------------------------------------------------------------------------------------
% ---------------------------System Setup------------------------------------
\section{System Models}    \label{sec_system_setup}

%%================================================================================================
%%================================================================================================
%%================================================================================================
In this section, we provide  a mathematical description for the  mmWave  cellular environment and system model under consideration.

\begin{table}[t]
\footnotesize  %\small% or footnotesize, scriptsize, tiny, etc.
  \caption{List of Main Variables and Their Physical Meanings}
  \centering
  \begin{tabular}{l|l}
\hline
\textbf{Variable}&
\textbf{Description} \\
\hline
\hline
$G_{\text{main}}$ &
Mainlobe beam gain \\
\hline
$G_{\text{side}}$ & Sidelobe beam gain \\
\hline
$H_{\text{BS},b}$&
Height of BS $b$\\
\hline
$H_{\text{UE}}$&
Height of UEs\\
\hline
$I_{i,g}$&
Interference power from BS $i$ to UE $g$\\
\hline
$L_{\text{grd}}$&
Length of square grid\\
\hline
${n}_b\left(\bx_b \right)$ &
Number of active UEs  in cell $b$ with area $\sum_{g} \!x_{b,g}L_{\text{grd}}^2$\\
\hline
$N_{\text{RF}}$ &
Number of RF chains at each BS \\ %(\tcr{Specify its unit, W or dBm?}) \\
\hline
$P_{b,g}$&
Transmit power of link from BS $b$ to UE $g$\\
\hline
 $p_{b,g}^{\text{blk}}$ &
Physical blockage probability of link from BS $b$ to UE $g$ \\
\hline
 $p_{b,g}^{\text{out}}$ &
Outage probability of link from BS $b$ to UE $g$ \\
\hline
$\text{PL}_{b,g}$&
 Pathloss of the link from BS $b$ to grid $g$\\
\hline
$R^{\text{max}}$ &
 Maximum allowed link distance for reliable communications \\
\hline
$r_{b,g}$&
Distance of the link from BS $b$ to grid $g$  \\
\hline
$r_{b}^{\text{max}}$&
Maximum link distace in cell $b$  \\
\hline
$\text{SINR}_{b,g}$&
SINR of link from BS $i$ to UE $g$\\
\hline
$\mathbf{X}\!\in\!\mathbb{B}^{B\!\times \!G}$ &
Association matrix\\
\hline
$\mathbf{x}_b\!\in\!\mathbb{B}^{1\!\times \!G}$ &
The $b$th row of $\mathbf{X}$\\
\hline
$\mathbf{y}\!\in\!\mathbb{B}^{1\!\times \!B}$ &
BS deployment vector\\
\hline
$z$&
SINR threshold\\
\hline
$\lambda_{\text{UE},g}$
&
UE density in grid $g$, \\
\hline
$\zeta$ &
UE outage tolerance   \\
\hline
 $\gamma$
& UE access-limited blockage   tolerance \\
\hline
 $\sigma^2$&
Noise power  \\
\hline \hline
\end{tabular}
\label{tab:notation}
\end{table}

%%================================================================================================
%%================================================================================================
% ---------------------------System Setup------------------------------------
\subsection{Urban Geometry}
%%================================================================================================
%%================================================================================================
\tcg{We consider a 3D urban geometry, for example, as illustrated in Fig.~\ref{fig:geometry}, consisting of buildings and streets.
In Fig.~\ref{fig:geometry}, mmWave BSs are mounted  on  walls of the buildings to serve outdoor active UEs on the streets in the downlink.}
\tcg{We assume that the  candidate BS locations are predetermined as red dots in Fig.~\ref{fig:geometry} and let the indices of the candidate BS locations be  $\mathcal{B}\!=\!\{1,2,\ldots,B\}$. Each candidate location has the height $H_{\text{BS},b},\forall b\in\mathcal{B}$.
If a BS is installed at the $b$th ($b\in\cB$) location,  $y_b=1$ and otherwise, $y_b=0$, where $y_b\in \mathbb{B}=\{0,1\}$ is the $b$th entry of the BS deployment vector $\mathbf{y}=[y_1,\ldots, y_B] \in \B^{1\times B}$.
Hereafter, a BS deployed at the $b$th candidate location is called ``BS $b$".}

Each mmWave BS has its coverage area, called a cell. Each BS serves outdoor active UEs inside its cell.
Due to the physical blockage by obstacles and buildings, the shape of a cell is irregular.
In our approach, to capture the irregularity of each cell, the whole outdoor area is divided into $G$ square grids in Fig.~\ref{fig:geometry}.
The location of each grid is represented by its center point. We let the index set of the grids be $\mathcal{G}\!=\!\{1,2,\ldots,G\}$. 
For simplicity,
we call a UE in grid $g$ as ``UE~$g$".
\tcg{An association indicator} $x_{b,g}\in \B$ is introduced, where  $x_{b,g}\!=\!1$ if the grid $g\in\cG$ resides in the cell of BS $b$, and $x_{b,g}=0$, otherwise. 
The cell area associated with BS $b$ then equals to $\sum_{g} x_{b,g} L_{\text{grd}}^2$, where $L_{\text{grd}}$ is the length of the side of a square grid in Fig. \ref{fig:geometry}.
The overlap between cells is allowed so that some UEs in a cell can be simultaneously covered by more than one BS (i.e., macro diversity)
\beq 
\sum_{b\in \cB} x_{b,g} \geq 1 . \label{macro_diversity}
\eeq
All \tcg{association  indicators} $\{ x_{b,g} \}$ are collected into \tcg{an association  matrix} $\mathbf{X}\in\mathbb{B}^{B\times G}$, where its $b$th row and $g$th column entry is $x_{b,g}$.
We assume that UEs have the same height $H_{\text{UE}}$ and  $H_{\text{UE}} \!< \! H_{\text{BS},b}, \forall b\in\mathcal{B}$.
\tcg{UEs in the outdoor area are randomly distributed and their placements can be accurately modeled by a non-homogeneous Poisson point process (PPP) with a certain UE density distribution \cite{NonPPP1,NonPPP2}.
For ease of exposition, when the partitioned square grid in Fig.~\ref{fig:geometry} is small enough, we assume that the users in a grid $g$ are uniformly distributed with a density $\lambda_{\text{UE},g}$, $g\!\in\! \cG$. The different grid can have different UE density $\lambda_{\text{UE},g}$.}
\tcg{This approximation becomes
exact as the square grid reduces to a point.}
\begin{figure}
  \centering
  \includegraphics[width=0.48\textwidth]{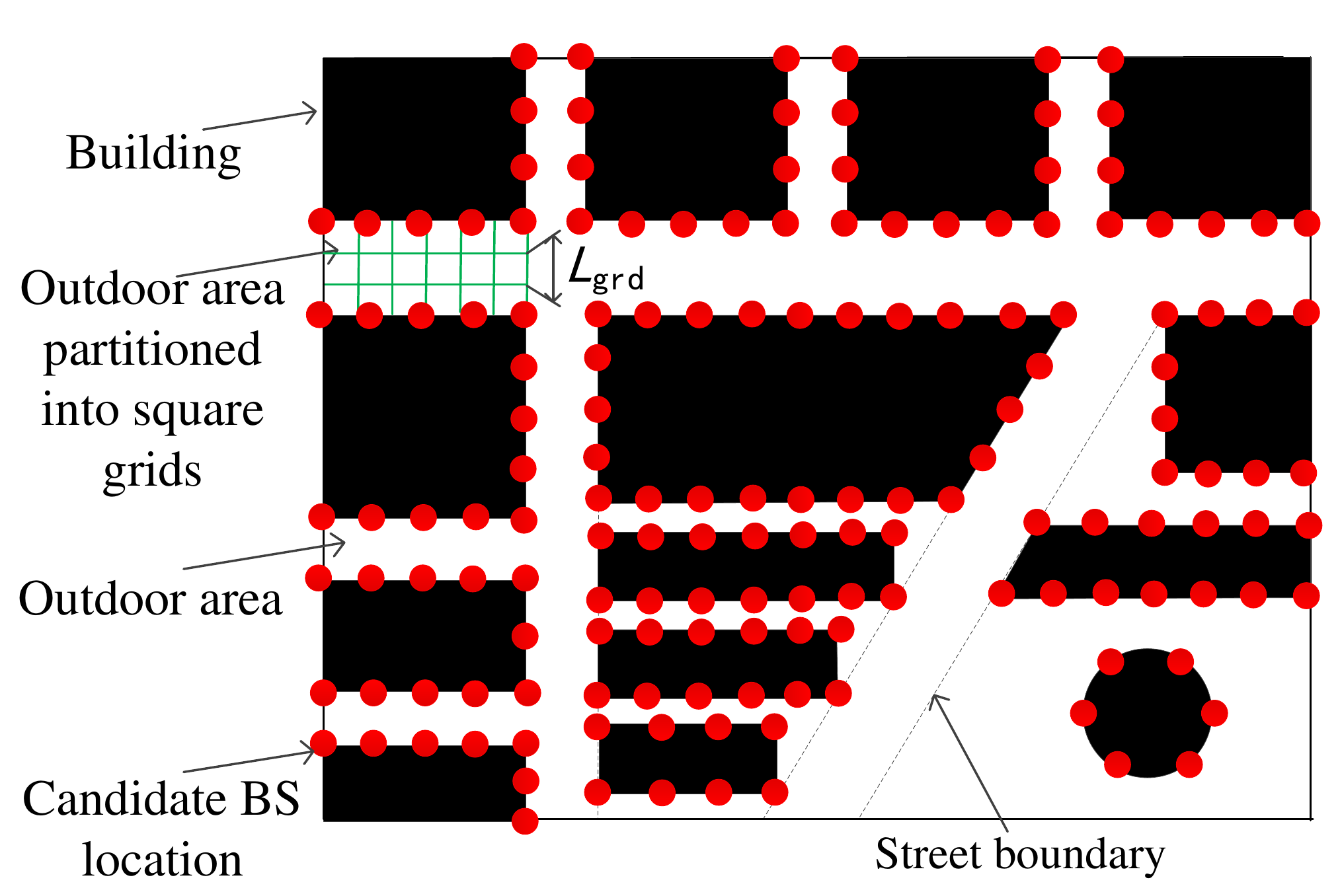}
  %\vspace{-1.5cm}
  \caption{\tcg{Bird's-eye view of an exemplary urban street geometry for mmWave BS deployment.}}
  \label{fig:geometry}
\end{figure}

%%%================================================================================================
%%%================================================================================================
%% ---------------------------System Setup------------------------------------
\subsection{Pathloss Model}   \label{Link_model}
%%%================================================================================================
%%%================================================================================================

\tcg{MmWave propagation experiencing severe pathloss and atmospheric impairments exhibits sparse multi-path channels.
In particular, weak diffraction and penetration at mmWave frequencies make the non-LoS (NLoS) paths suffer from much severer attenuation than the LoS paths.
As a result, the quality of  a mmWave link is dominantly determined by its LoS link \cite{MyRay_tracing}. Hence, in this work, we mainly focus on the LoS path.
In the recent 3GPP specifications (Release 15 \cite{3GPP38211}), the $28$ GHz mmWave band has been considered as one of the standard frequencies for 5G cellular communications.  Throughout the paper, we will use the $28$ GHz LoS pathloss model based on the measurement campaign conducted in an urban area  \cite{Link_status}, which is characterized by
\begin{eqnarray}
 \text{PL}_{b,g}= 10^{-3.24 - 2.1 \cdot \log_{10}(r_{b,g})- 2.0 \cdot \log_{10}(28)},
\label{pathloss}
\end{eqnarray}
where $\text{PL}_{b,g}$ is the pathloss  from BS $b$ to a UE  $g$, and $r_{b,g}$ is the link distance in meters.}
\tcg{We assume that the link distance $r_{b,g}$ is limited by a constant $R^{\text{max}}$,  $r_{b,g}\leq R^{\text{max}}$, where $R^{\text{max}}$ is the maximum allowed mmWave link distance for reliable communications.\footnote{Based on the outdoor mmWave measurements in \cite{Link_distance}, for example, a $200$ meter LoS link in the range of $28$GHz system was measured to be extremely unreliable for communication, thereby $R^{\text{max}}=200$.} Thus, we assume any link with the distance $r_{b,g}> R^{\text{max}}$ is in outage.}

%%%================================================================================================
\subsection{Hybrid Array and Beam  Pattern}   \label{Beampattern}
%%%================================================================================================
%%%================================================================================================

\subsubsection{Hybrid Analog-Digital Arrays}
The mmWave frequencies force systems to use large-sized antenna arrays to generate highly directional narrow beams in order to overcome the propagation impairments \cite{Kims_gen}.
To make the large-sized arrays available at low cost, analog arrays that are driven by a limited number ($N_{\text{RF}}$)  of RF chains are typically employed.
This system is referred to as a hybrid analog-digital mmWave MIMO system \cite{HealmmWave,Hadi2016,BSCap_limit2, Zhang18,aych_RFChain,Alkhateeb}, which is the major realization technology of 5G mmWave BS systems. Throughout the paper, we assume such hybrid mmWave MIMO BS systems.

%\taejoon{I feel this subsubsection is redundant. We already have much detailed description in the next section. This is an example of unorganized writing. This should be removed.} \tcgr{\subsubsection{UE Access-Limited Blockage}... %Unlike fully digital MIMO systems that have been largely deployed at sub-$6$GHz cellular frequencies, the mmWave-based system suffers from \tcg{UE access-limited} blockage \cite{ZhangYue}.
%More specifically, the maximum number of data streams (beams) that a BS can simultaneously generate is strictly limited by the number of RF chains $N_{\text{RF}}$ in the hybrid MIMO system.
%When the number of active UEs in a cell exceeds the number of RF chains $N_{\text{RF}}$, not all active UEs can be served and some UEs are in the UE access-limited blockage. We will mathematically model the UE access-limited blockage in Section \ref{CapacityCons}.
%}
%where the maximum number of data streams or beams that a BS can simultaneously generate is strictly limited by the number of RF chains $N_{\text{RF}}$.
%When the number of active UEs exceeds $N_{\text{RF}}$ in a cell, not all active UEs can be served and some UEs are in capacity-limited blockage.

\subsubsection{Beam Pattern}
For ease of exposition, we assume a simplified directional beam pattern, also called the cone-shaped beam pattern, at BSs as follows
\begin{equation}
\begin{small}
\textstyle {G_{\text{BS}}(\theta,\phi) \!=\!}
\textstyle \left\{
\begin{array}{ll}
\textstyle \!\!\!G_{\text{main}},                 \!& \text{if}\; |\theta| \leq \frac{\Delta_{\theta}}{2},\;   |\phi| \leq \frac{\Delta_{\phi}}{2} \\
\!\!\! G_{\text{side}},   \!\!& \textstyle \text{otherwise}
\end{array},\label{BeamGain}
\right.
\end{small}
\end{equation}
where $G_{\text{main}}$ and $G_{\text{side}}$ are the array gains at the mainlode and sidelobe, respectively. Then, $G_{\text{BS}}(\theta,\phi)$ represents the beam gain at elevation $\theta$ and azimuth $\phi$ directions.
The $\Delta_{\theta}$ and $\Delta_{\phi}$ in \eqref{BeamGain} are, respectively, the beamwidths at the elevation $\theta$ and azimuth $\phi$ directions.
Though model (\ref{BeamGain}) is an approximation, it can be shown that \eqref{BeamGain} closely approach to the actual beam pattern as the array size increases \cite{Mac_div2, RobertHealth2,BeamPattern}.
The BS aligns its beam to provide  beamforming gain $G_{\text{main}}$ to a UE.
In this work, we assume a single antenna UE that generates an omnidirectional receive beam.

%\taejoon{The next two sentences sounds awkward and they are not contextually aligned with this subsection, whose title is ``Urban Geometry". In this subsection, we only discuss Urban Geometry and its model, not UE outage tolerance and about how it is determined. This needs to be moved where we first define the tolerance variables} * At the stage of network deployment, we focus on the BS deployment to guarantee the successful channel access of randomly distributed UEs in the considered area with a required UE outage tolerance.
%(\tcr{Note that it is wired to define $\zeta$ here even without mathematically defining UE outage. I suppose $\zeta$ must be introduced after we formulate the UE outage in Section III.B}).*}
%%%================================================================================================
%%================================================================================================
%%================================================================================================
%%================================================================================================
%%================================================================================================
% \section{BS Deployment Approaches}~\label{sec:two_methods}
\section{Link Analysis and Problem Formulation}   \label{sec_ProblemForm}
%%================================================================================================
%%================================================================================================
%%================================================================================================
%%================================================================================================
%%%================================================================================================
\tcg{In this section, we identify several key constraints and the statistics of elementary events, including  the physical blockage probability, BS association constraints, UE access-limited blockage probability, and SINR outage probability that are used to formulate the UE outage constraint. 
These are then combined to define the mmWave BS deployment problem at the end of this section.}

%%%================================================================================================
%%%================================================================================================
\subsection{Physical Blockage Probability} \label{PhyBlock}
%%%================================================================================================
%%%================================================================================================
The physical blockage of mmWave links depends not only on the distance of a link, but also on the density and sizes of random obstacles \cite{twoBall_model, Mac_div2}.
%In this work, 3D ray tracing \cite{MyPIMRC} is adopted to extract the LoS path from each BS $b$ to each grid $g$, \tcg{while limiting the path distance $r_{b,g}$ to $R^{\text{max}}$ meters} ($r_{b,g}\leq R^{\text{max}}$). We then model random blockage effects and derive the link blockage probability as follows.
%\subsubsection{LoS Path Blockage Probability}
The random obstacles are assumed to be impenetrable cubes, based on the Boolean scheme, and the placement of each obstacle follows a homogeneous PPP with the density $\lambda_{\text{obs}}$ \cite{RobertHealth1}.
The physical blockage probability  with the path length  $r_{b,g}$   is then given by 
\begin{equation}
\d4  p_{b,g}^{\text{blk}}  \!=\!
\textstyle \!\left\{
\begin{array}{ll}
\textstyle \!\!1\!-\!\exp(-\beta r_{b,g}\!-\!\alpha)     \!\!\! \!\!&,~\text{if $\exists$ LoS path} \\
 \!\!1   \!\!\!&,~ \textstyle \text{otherwise}
\end{array},   \hspace{-0.3cm}
\right.
\label{LoS_Blockage}
\end{equation}
where $r_{b,g}\leq R^{\text{max}}$ and \tcg{$\alpha$ and $\beta$ are parameters that depend on the density and sizes of the obstacles \cite{RobertHealth1}. \tcg{The variables in \eqref{LoS_Blockage} follow $\beta\!=\!2\lambda_{\text{obs}} \frac{ E[L_{\text{obs}}] + E[W_{\text{obs}}] }{\pi}\eta$ and $\alpha =\lambda_{\text{obs}} E[L_{\text{obs}}]E[W_{\text{obs}}]$, where $E[L_{\text{obs}}]$ and $E[W_{\text{obs}}]$ are the expected length and width of obstacles, respectively, and $\eta=1\!-\!\int_{0}^1\int_{0}^{sH_{\text{UE}}+(1-s)H_{\text{BS}}} f_h(x) dxds$ where $f_h(x)$ is the probability density function of the height of an obstacle.}
The intuition behind  \eqref{LoS_Blockage} is that the denser the obstacle distribution and the larger the obstacle sizes, the larger $\alpha$ and $\beta$ values will be, resulting in a higher blockage probability.}

%%%================================================================================================
%%%================================================================================================
\subsection{BS Association Constraints }  \label{cell_coverage}
%%%================================================================================================
%%%================================================================================================
 \tcg{In this subsection, we list the rules for elements in the association matrix $\bX$. These elements will be used for formulating the UE access-limited blockage and SINR outage probabilities in the next two subsections.}

If a BS is installed at the $b$th candidate location ($y_b\!=\!1$), the association variable is either $x_{b,g}=1$ or $x_{b,g}=0$, while if  $y_b\!=\!0$, then $x_{b,g}\!=\!0$, yielding
\begin{eqnarray}
x_{b,g}\leq  y_{b}, \quad  b\in\mathcal{B},\;  g\in \mathcal{G}.
\label{xy_relation}
\end{eqnarray}
\tcg{Based on the physical blockage model in the previous subsection, it is evident that necessary conditions for $x_{b,g}=1$ are: (i) the link distance $r_{b,g}\leq R^{\text{max}}$ and (ii) the physical blockage probability $p_{b,g}^{\text{blk}}<1$ in \eqref{LoS_Blockage}, which leads to
\beq
x_{b,g} \leq \mathbb{I}_{\left\{{{ r_{b,g}\leq R^{\text{max}}, \;p_{b,g}^{\text{blk}}}<1 } \right \}}, \; b\in\mathcal{B}, g\in\mathcal{G},
\label{LoS_visible_cons}
\eeq
where $\mathbb{I}_{ \left\{ \cA \right\} }$ is an indicator function: $\mathbb{I}_{ \left\{ \cA \right\} } = 1$ if the event $\cA$ is true, and $\mathbb{I}_{ \left\{ \cA \right\} }=0$ otherwise.
We assume that the grids closer to a BS have the priority to be associated with the BS. This is to say, whenever a grid $g$ is served by BS $b$ ($x_{b,g}=1$), other grid $s\in\cG$ with $r_{b,s}\leq r_{b,g}$ and $p_{b,s}^{\text{blk}}<1$ should also be served by the BS $b$, leading to
\beq
\mathbb{I}_{\left\{{{ r_{b,s}\leq r_{b,g}, \;p_{b,s}^{\text{blk}}}<1  } \right \}}x_{b,g} \leq x_{b,s}, \; b\in\mathcal{B}, g, s\in\mathcal{G}.
\label{LoS_visible_cons2}
\eeq}

\tcg{
The conditions in \eqref{xy_relation}, \eqref{LoS_visible_cons}, and \eqref{LoS_visible_cons2} will be  incorporated in our proposed BS deployment problem as the BS association constraints.}

%%%================================================================================================
%%%================================================================================================
\subsection{\tcg{UE Access-Limited Blockage Probability}}    \label{CapacityCons}
%%%================================================================================================
%%%================================================================================================

\tcg{When the mmWave channel between a UE and BS does not experience physical blockage, the UE can attempt channel access to the BS.  However, it can still be blocked due to saturated UE access within a cell.
More specifically, the maximum number of UEs that a BS can simultaneously serve is strictly limited by the number of RF chains $N_{\text{RF}}$ in hybrid MIMO BS systems.
UE access-limited blockage occurs whenever the number of active UEs in a cell is larger than $N_{\text{RF}}$.
This blockage has an intuitive interrelation between the UE density $\lambda_{\text{UE},g}$ and the number of grids covered by a BS $\sum_g x_{b,g}$.
For instance, if a BS covers a larger number of grids, it results in a higher probability of UE access-limited blockage. 
The same is true when the UE density per grid $\lambda_{\text{UE},g}$ grows.
To capture this interrelation and to use it for controlling the numbers of grids covered by BSs, we now derive the UE access-limited blockage probability between BS $b$ and a UE in grid $g$ denoted by $\rho_{b, g}(\bx_b)$, where $\bx_b = [x_{b,1}, \ldots, x_{b,G}]\in \B^{1 \times B}$ is the $b$th row of the association matrix $\bX$.}

%%%%% % %%%%%% Motivated by the above bound, we focus on the constraint $\rho_b \leq \gamma$ and identify a closed-form of it. 

\tcg{We let ${n}_b\left(\bx_b \right)$ be the number of active UEs without physical blockage in the cell area $\sum_{g} \!x_{b,g}L_{\text{grd}}^2$ of BS $b$. 
Then, the  UE access-limited blockage at BS $b$ occurs when ${n}_b\left(\bx_b \right)>N_{\text{RF}}$, in which  ${n}_b\left(\bx_b \right)$ is a random variable that depends on UE distribution and random physical blockage. Assume that each of the  ${n}_b\left(\bx_b \right)$ UEs has the equal probability to have successful channel access without UE access-limited blockage. For a given ${n}_b\left(\bx_b \right)$ with  ${n}_b\left(\bx_b \right)>N_{\text{RF}}$, the UE $g$ is then in UE access-limited blockage with probability 
$$\frac{{n}_b\left(\bx_b \right)-N_{\text{RF}}}{{n}_b\left(\bx_b \right)}.$$
The UE access-limited blockage probability $\rho_{b, g}(\bx_b) $ is therefore given by  
\beq
\rho_{b, g}(\bx_b) = \sum\limits_{\forall x\; \text{with}  \;x>N_{\text{RF}}}  \Pr({n}_b\left(\bx_b \right)=x)\frac{x-N_{\text{RF}}}{x}
\label{eqn:UeAccessLimitedBlockPro1}
\eeq
}

%, leading to 
%\begin{eqnarray}
%\rho_b = \Pr\left(n_b(\bx_b)> N_{\text{RF}} \right) %\leq \gamma.
%\label{BSCap_ProsingleUE}
%\end{eqnarray}
%\taejoon{I just remove the paragraph discussing the relation $ \rho_{b,g}\leq \rho_{b}\leq  \gamma$, which I think unnecessary and disturbing the flow of derivation. If you have any concern with it please let me know.} \tcgr{It should be mentioned...}
%that  \eqref{BSCap_ProsingleUE} is the constraint for BS. When it occurs at the BS, UEs associated with the BS have some probability to be selected for data transmission. Hence,the user access-limited blockage probability $\rho_{b,g}$  at a UE in grid $g$ satisfies $$ \rho_{b,g}\leq \rho_{b}\leq  \gamma$$ By setting $\gamma$ to be small, e.g., $\gamma=0.05$, we approximately have $\rho_{b,g}\approx \rho_{b}$ due to $\left|  \rho_{b,g}- \rho_{b} \right|\leq  \gamma$. In what follows, we use the $\rho_{b}$ in \eqref{BSCap_ProsingleUE} to formulate the user access-limited blockage constraint. It is clear that it guarantees the $\rho_{b,g}\leq \gamma$.
%In what follows, we drive a closed-form for the constraint $\Pr \left(n_b(\bx_b)\!>\! N_{\text{RF}} \right)\leq \gamma$ in (\ref{BSCap_ProsingleUE}).

Identifying  $\Pr({n}_b\left(\bx_b \right)=x)$ in \eqref{eqn:UeAccessLimitedBlockPro1}
requires the distribution of $n_b(\mathbf{\bx_b})$.
By leveraging the independent thinning property of PPP \cite{thinning}, the number of active UEs, associated with BS $b$, without physical blockage per grid $g$ is Poisson distributed with the mean $\lambda_{\text{UE}, g}L^2_{\text{grd}}(1-p_{b,g}^{\text{blk}}).$ %(\tcr{Please check in this formula if  $p_{b,g}^{\text{blk}}$ or $1-p_{b,g}^{\text{blk}}$ is correct? $\quad$ Miaomiao: thanks for your comments. The correct one should be $1-p_{b,g}^{\text{blk}}$})
Because the sum of independent Poisson random variables is still Poisson,  ${n}_{b}(\bx_b)$ is Poission-distributed with the mean
\begin{eqnarray}
 E[{n}_{b}(\bx_b)]=\sum_{g\in\mathcal{G}} x_{b,g}\lambda_{\text{UE},g}L^2_{\text{grd}}(1-p_{b,g}^{\text{blk}}).
\label{Mean_k_b}
\end{eqnarray}
The closed-form expression of (\ref{eqn:UeAccessLimitedBlockPro1}) is therefore given by
\begin{eqnarray}
\rho_{b, g}(\bx_b) \!=\!\sum_{i=N_{\text{RF}}+1}^{+\infty}  \! \! \frac{E[{n}_b(\bx_b)]^{i}}{i!}e^{-E[{n}_b(\bx_b)]}\frac{i-N_{\text{RF}}}{i} .
\label{express1}
\end{eqnarray}

The following lemma characterizes the relationship between the UE access-limited blockage probability in \eqref{express1} and the cell coverage and UE density.

%is complicate, and it is not straightforward to  observe the relationship between the $\rho_{b, g}(\bx_b)$ and cell area and/or UE density stated previously. Hence, we  introduce the following lemma to demonstrate it.  

\begin{lemma}
\label{lemma2}
\tcg{The UE access-limited blockage probability $\rho_{b, g}(\bx_b)$ in \eqref{express1} is a monotonically increasing function of the cell coverage (i.e., $\sum_{g}x_{b,g}$) and/or UE density $\lambda_{\text{UE},g}$.} 

%increases as the BS $b$ covers more grids (i.e., $\sum_{g}x_{b,g}$ increases) and/or UE density $\lambda_{\text{UE},g}$ grows. 
\end{lemma}
\begin{proof}
\tcg{It is not difficult to observe from  \eqref{Mean_k_b} that  $E[{n}_{b}(\bx_b)]$ increases as the $\sum_{g}x_{b,g}$ and/or UE density $\lambda_{\text{UE},g}$ grow. Hence, the proof of the lemma boils down to showing that the $\rho_{b, g}$ in \eqref{express1} is a monotonically increasing function of  $E[{n}_b(\bx_b)]$. 
This can be verified by taking the first-order derivative of $\rho_{b, g}(\bx_b)$ with respect to $E[{n}_b(\bx_b)]$, yielding}
\begin{eqnarray} \nonumber
\textstyle
\frac{\partial \rho_{b, g}(\bx_b)}{\partial E[{n}_b(\mathbf{x}_b)]}
\!\!\d4&\= &\d4\!\! \textstyle \sum\limits_{i=N_{\text{RF}}+1}^{+\infty} \frac{ E[{n}_b(\mathbf{x}_b)]^{i-1}}{(i-1)!}e^{-E[{n}_b(\mathbf{x}_b)]}\frac{i-N_{\text{RF}}}{i}   \\    \nonumber 
\d4 &  &\d4 \hspace{1cm}- \! \textstyle \sum\limits_{i=N_{\text{RF}}+1}^{+\infty} \frac{E[{n}_b(\mathbf{x}_b)]^{i}}{i!}e^{-E[{n}_b(\mathbf{x}_b)]}\frac{i-N_{\text{RF}}}{i} \\ \nonumber
\d4 &\stackrel{(a)}\=&  \!\!\d4 \textstyle  e^{-E[{n}_b(\mathbf{x}_b)]}\!\! \left(\!\! \frac{E[{n}_b(\mathbf{x}_b)]^{N_{\text{RF}}}}{(N_{\text{RF}}+1)!} \!+\!\!\! \d4  \sum\limits_{i=N_{\text{RF}}+1}^{+\infty} \!\!\d4 \frac{E[{n}_b(\mathbf{x}_b)]^{i}}{i!} \!\left(\!\frac{N_{\text{RF}}}{i}\!\!-\!\!\frac{N_{\text{RF}}}{i+1} \!\right) \!\!\right) \;\;\\ \nonumber
&\geq& 0,
\end{eqnarray}
\tcg{where the step (a) follows from the fact that  $\sum\limits_{i=N_{\text{RF}}+1}^{+\infty} \!\!\!\! \frac{ E[{n}_b(\mathbf{x}_b)]^{i-1}}{(i-1)!}e^{-E[{n}_b(\mathbf{x}_b)]}\frac{i-N_{\text{RF}}}{i}$ in the first equality can be rewritten as}
$$\textstyle   e^{-E[{n}_b(\mathbf{x}_b)]} \frac{E[{n}_b(\mathbf{x}_b)]^{N_{\text{RF}}}}{(N_{\text{RF}}+1)!}  + \!\! \sum\limits_{i=N_{\text{RF}}+1}^{+\infty} \!\frac{ E[{n}_b(\mathbf{x}_b)]^{i}}{i!}e^{-E[{n}_b(\mathbf{x}_b)]}\frac{i+1-N_{\text{RF}}}{i+1}.  $$
This completes the proof.
\end{proof}

%\begin{remark}
%\label{remark:BS_capacity}
%Lemma \ref{lemma2} provides useful insights into the number of grids covered by a BS in relation to the UE access-limited blockage.
%For example, the larger the tolerance $\gamma$ in  (\ref{BSCap_ProsingleUE}), the more grids a BS can associate due to a larger $\Phi$ value in (\ref{Re_NumUE_PerBS}).
%Moreover, seen from \eqref{Re_NumUE_PerBS}, increasing the UE density $\lambda_{\text{UE},g}$ will reduce the number of grids covered by a BS, which makes sense because a large $\lambda_{\text{UE},g}$ will lead to severe UE access-limited blockage.
%For a fixed $\gamma$ (i.e., fixed   $\Phi$) and fixed $\lambda_{\text{UE},g}, \forall g\in\cG$,  the number of grids  associated with BS $b$ is limited by the UE access-limited blockage constraint in \eqref{Re_NumUE_PerBS}.
%\end{remark}

%%%================================================================================================
%%%================================================================================================
%%%================================================================================================
\subsection{SINR Outage Probability}    \label{SINRConst}
%%%================================================================================================
%%%================================================================================================
%%%================================================================================================
%Suppose the link from BS $b$ to grid $g$ is in connectivity with state $q\in\{\text{LoS, NLoS}\}$.
%In this subsection, we aim to find a tractable upper bound of the
%$\Pr\left(\text{SINR}^{q}_{b,g}< z| \mathbf{y},\mathbf{X}\right)$ in \eqref{SumPro_23}.

While a UE that does not experience physical and UE access-limited blockages can acquire initial access to BS, the acquired link can be unavailable due to intercell and intracell interference from surrounding BSs.
Hence, describing the SINR outage $\Pr\left(\text{SINR}_{b,g}(\by, \bX)\!<\! z| \mathbf{y},\mathbf{X}\!\right)$, where  $\text{SINR}_{b,g}(\by, \bX)$ is the SINR of a link from BS $b$ to a UE  $g$ and $z$ is the SINR threshold for reliable communications, is of interest. 
Directly analyzing the distribution of $\text{SINR}_{b,g}(\mathbf{y},\mathbf{X})$ in the mmWave environment is very difficult due to stochastic physical blockage and UE distribution.
In this subsection, we resort to a deterministic lower bound of $\text{SINR}_{b,g}$ to propose a closed-form approximate of $\Pr\left(\text{SINR}_{b,g}(\by, \bX)< z| \mathbf{y},\mathbf{X}\right)$.
%This will contribute to yielding a UE outage upper bound in Section~\ref{UeOutageCons}. %Whenever the UE outage upper bound is limited by a tolerance, the real  UE outage is bounded as well.

% (i.e., upper bound of the $p_{b,g}^{\text{SINR}}(\mathbf{y}, \mathbf{X})$ in \eqref{Pro3}), which readily characterizes a sufficient condition of \eqref{SINR_constraint}.

We assume an equal power allocation per UE and write the desired signal power $P_{b,g}(\mathbf{x}_b)$ received at an active UE in grid $g$ from its serving BS $b$ as
\begin{eqnarray} \nonumber
  P_{b,g}(\mathbf{x}_b)\d4\d4&&=\frac{ x_{b,g} P_{\text{TX}} }{\min(n_b(\bx_b), N_{\text{RF}})}  G_{\text{main}} \text{PL}_{b,g}  \\
                         && \geq \frac{x_{b,g} P_{\text{TX}} }{N_{\text{RF}}}  G_{\text{main}} \text{PL}_{b,g} \;\triangleq \;\overline{P}_{b,g}( x_{b,g}),
\label{desirepower}
\end{eqnarray}
where $P_{\text{TX}}$ is the total BS transmit power, $\min(n_b(\bx_b), N_{\text{RF}})$ is the number of served active UEs by BS $b$,
$G_{\text{main}}$ follows (\ref{BeamGain}), and  $\text{PL}_{b,g}$ is in \eqref{pathloss}.
The last inequality in \eqref{desirepower} is due to $\min(n_b(\bx_b), N_{\text{RF}}) \leq N_{\text{RF}}$.
We now capture the composite link interference power under an assumption that the mainlobe of a 3D beam in \eqref{BeamGain} is perfectly aligned with the intended UE and is narrow enough not to cause interference to unintended UEs.
Thus, it is the sidelobe of the beam that causes interference with probability $1$.
A BS $i$ serving $\min(n_i(\bx_i), N_{\text{RF}})$ UEs has in total $\min(n_i(\bx_i), N_{\text{RF}})$ beams and can possibly impose interference $I_{i,g}(\mathbf{x}_b)$ to the UE in grid $g$ with $\min(n_i(\bx_i), N_{\text{RF}})-x_{i,g}$  interfering sidelobes
\tcg{if there exists an LoS path between the BS $i$ and UE $g$.}
In the case of $p_{i,g}^{\text{blk}}<1$ in (\ref{LoS_Blockage}), the interference $I_{i,g}(\mathbf{x}_b)$  is a Bernoulli random variable with its value either $0$ (blocked) or positive (unblocked), yielding
\begin{eqnarray} \nonumber
\label{Single_InterferUB}
I_{i,g}(\mathbf{x}_b)\d4\d4\d4 && =\!  \mathbb{I}_{ \left\{ \cA_{i,g}\right\} } \!\frac{\big(\!\min(n_i(\bx_i), N_{\text{RF}})\!-\!x_{i,g}\big)P_{\text{TX}}}{\min(n_i(\bx_i), N_{\text{RF}})} G_{\text{side}} \text{PL}_{i,g} \;\;\quad  \\  
%\nonumber &&\leq   \! \left(1\!-\!\frac{x_{i,g}}{N_{\text{RF}}}\right)P_{\text{TX}}G_{\text{side}} \text{PL}_{i,g}  \\    
\d4\d4 &&\leq\! \left(\!1\!-\!\frac{x_{i,g}}{N_{\text{RF}}}\right)P_{\text{TX}}G_{\text{side}} \text{PL}_{i,g} \!\triangleq \! \widehat{I}_{i,g}(x_{i,g}),
\label{single_inf}
\end{eqnarray}
%\taejoon{I changed these equation because I thought there were several errors. I suggest checking them carefully.}
where $\cA_{i,g}$ denotes the event that the LoS path  between  the  BS  $i$ and  UE $g$ is not physically blocked. 
%is an indicator function: $\mathbb{I}_{ \left\{ \cA \right\} } = 1$ if the event $\cA$ is true, and $\mathbb{I}_{ \left\{ \cA \right\} }=0$ otherwise. % \taejoon{Where is the indicator function in (12)? We should first define $I_{i,g}(\mathbf{x}_b)$ as closed-form in (12). Something is wrong here. Moreover, you should have some terms/term including $p_{i,g}^{\text{blk}}<1$ to make Remark 1 understandable. Note that (6) is the first place you define the indicator function. This is confusing.}. 
%The second inequality follows from the fact that $\min(n_i(\bx_i), N_{\text{RF}})\leq N_{\text{RF}}$ \taejoon{you already used this in (11). Move this to after (11) or remove.}, and 
\tcg{The last inequality in \eqref{single_inf} is due to the facts that $\mathbb{I}_{ \left\{ \cA_{i,g}\right\} } \leq 1$ and $\min(n_i(\bx_i), N_{\text{RF}})\leq N_{\text{RF}}$.  
When $n_i(\bx_i) = N_{\text{RF}}$, the  $\widehat{I}_{i,g}(x_{i,g})$ is the positive value of the Bernoulli random  $I_{i,g}(\mathbf{x}_b)$ and becomes tight if the link is blocked with a low probability.}
Accounting for (\ref{desirepower}) and (\ref{single_inf}), we have the lower bound  %$\text{SINR}_{b,g}(\mathbf{y},\mathbf{X}) \geq \overline{\text{SINR}}_{b,g}(\mathbf{y},\mathbf{X})$, where
\begin{eqnarray}
 \text{SINR}_{b,g}(\mathbf{y},\mathbf{X}) \geq \overline{\text{SINR}}_{b,g}(\mathbf{y},\mathbf{X}) \!=\!\frac{\overline{P}_{b,g}(x_{b,g}) }{\sigma^2 + \sum\limits_{i\in \mathcal{B}} y_i \widehat{I}_{i,g}(x_{i,g})}, 
\label{SINR_Exp}
\end{eqnarray}
%\taejoon{I suggest not to omit variables of a fuction for clarity. In above and below, I resored all $(\mathbf{y},\mathbf{X})$ omitted. Please do for all such functions and variables.}
where the $\sigma^2$ is the noise power. \tcg{The conditional probability of  $\overline{\text{SINR}}_{b,g}(\mathbf{y},\mathbf{X})<z$ given $\mathbf{y}$ and $\mathbf{X}$ is therefore given by
\begin{eqnarray}
\Pr\left(\overline{\text{SINR}}_{b,g}(\mathbf{y},\mathbf{X})\!<\! z| \mathbf{y},\mathbf{X}\!\right)=\I_{ \big\{ \overline{\text{SINR}}_{b,g}(\mathbf{y},\mathbf{X}) < z|\mathbf{y}, \mathbf{X} \big\} },
\label{eqn:pUB}
\end{eqnarray}
which is an upper bound of the SINR outage probability
\beq
\Pr\left({\text{SINR}}_{b,g}(\mathbf{y},\mathbf{X})\!<\! z| \mathbf{y},\mathbf{X}\!\right) \d4 &\leq& \d4  \Pr\left(\overline{\text{SINR}}_{b,g}(\mathbf{y},\mathbf{X})\!<\! z| \mathbf{y},\mathbf{X}\!\right). \nn \\
&& \hspace{3cm}
\label{SINRoutage_upperbound}
\eeq }
\vspace{-0.5cm}
\begin{remark}
\tcg{The lower bound in \eqref{SINR_Exp} is obtained based on the lower bound of the desired signal power in \eqref{desirepower} and the upper bound of the interference power in \eqref{single_inf}. 
Note that the bound in \eqref{desirepower} becomes tight when $n_b(\bx_b) \approx N_{\text{RF}}$; it indeed becomes the equality when $n_b(\bx_b) = N_{\text{RF}}$.   %The accumulated interference $\sum_{i\in \mathcal{B}} y_i \widehat{I}_{i,g}(x_{i,g})$ in \eqref{SINR_Exp} is dominated by  nearby BSs to the grid $g$.
Similarly, the bound in \eqref{single_inf} becomes tight when $n_b(\bx_b) \approx N_{\text{RF}}$. It can be further tighten when the accumulated interference $\sum_{i\in \mathcal{B}} y_i \widehat{I}_{i,g}(x_{i,g})$ in \eqref{SINR_Exp} is dominated by  nearby BSs of the grid $g$ because in this case the link distances between the nearby BSs and the grid $g$ are relatively small, resulting in low physical blockage probabilities $p_{i,g}^{\text{blk}}$ and revealing a higher chance for the event $\cA_{i,g}$ in \eqref{single_inf}.} 
\end{remark}

%\taejoon{Below remark basically destroy what we've ever built for SINR outage analysis in this paper because it said interference in the mmWave specturm is not an issue or negligible. By saying it is noise-limited, it kills this subsection. Moreover, most of the points are based on heuristic arguments. Don't create such artifacts. Reviewers would not follow it and backfire us. This being said, I removed the remark below, which is inappropriate. Be careful about saying the channel is noise-limited.}

%\taejoon{To address the reviewer's comment, I suggest to put more effort on remedying it by (i) explaining difficulty on verifying the tightness analytically (not sure if the tightness analysis is what the reviewer wanted), (ii) emphasize our focus and novel contributions, (iii) perform some numerical simulations showing the tightness of the bound and (iv) discuss a couple of scenarios when the bound becomes tight. Appropriate place will be the response letter (I don;t think it is a good idea to include this in the simulation section) and we may need to add a paragraph saying this investigation is difficult and studying tighter approximation is a future work or something like this. }

%%%================================================================================================
%%%================================================================================================
%% ---------------------------System Setup------------------------------------
\subsection{UE Outage Constraint}  \label{UeOutageCons}
%%%================================================================================================
%%%================================================================================================
To formulate the UE outage constraint that takes in the physical blockage, UE access-limited blockage, and SINR outage, we first identify the UE outage probability associated with a single link from a BS to a UE. This is then extended to the UE outage constraint that captures the effect of surrounding BSs.

\subsubsection{Single-Link UE Outage Probability}

\begin{figure}
  \centering
  \includegraphics[width=0.48\textwidth]{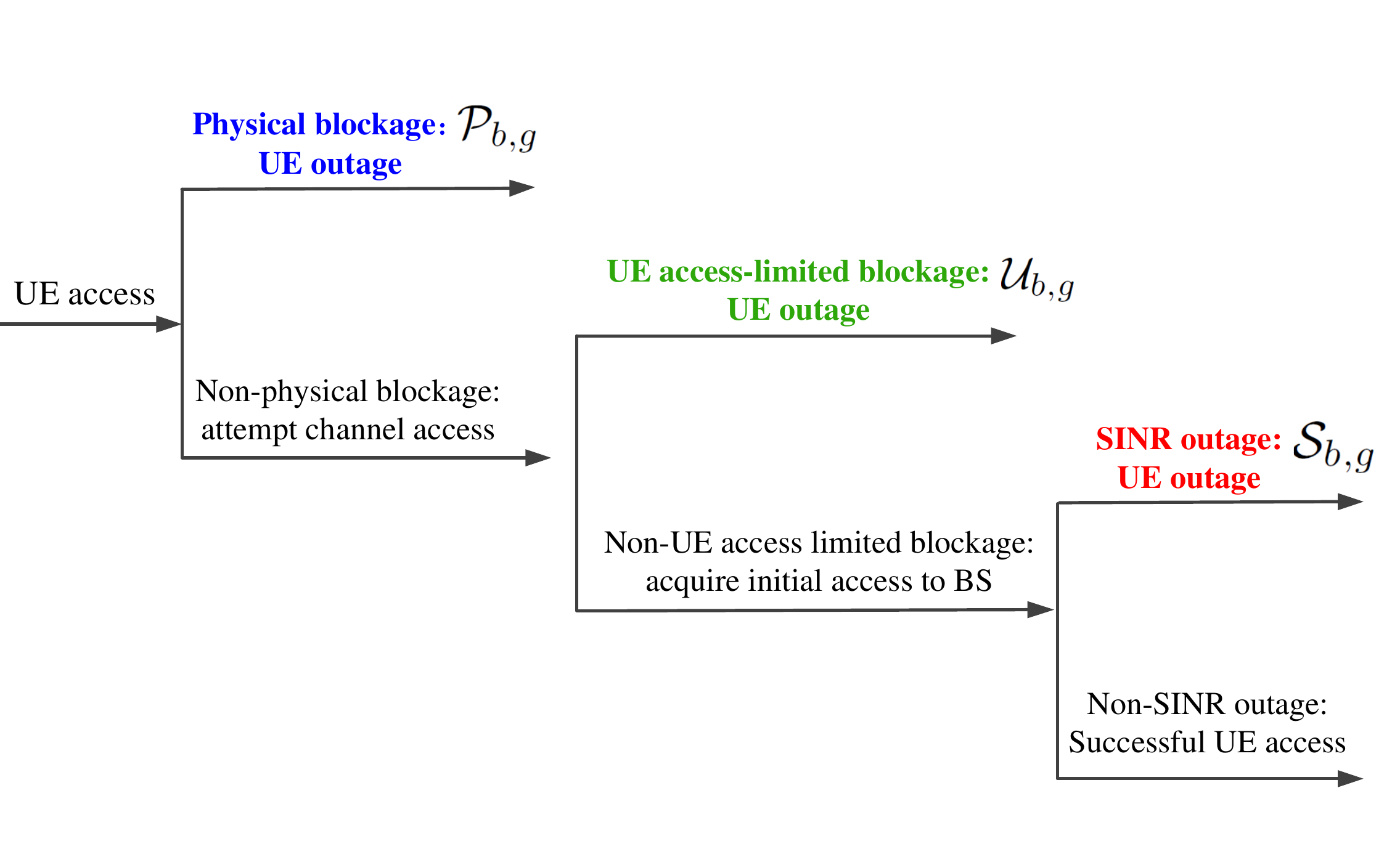}
  %\vspace{-1.5cm}
  \caption{Logical relationship of the three events causing UE outage. }
  \label{fig:ThreeEvents}
\end{figure}

\tcg{A link from BS~$b$ to a UE $g$ is in outage if one of the following mutually exclusive events occurs as described in the previous subsections: (i) Event $\cP_{b,g}$: the channel is physically blocked; (ii) Event $\cU_{b,g}$: the channel is physically unblocked, but the UE access-limited blockage occurs; and (iii) Event $\cS_{b,g}$: the link has no blockage and the UE acquires initial access to BS (i.e., SINR$_{b,g}\geq 0$), but the SINR outage occurs. The logical relationship of the three events  is  illustrated  in  Fig.  \ref{fig:ThreeEvents}. 
Because the three events are mutually exclusive, the single-link outage probability is given by
\beq
p_{b,g}^{\text{Out}} =\Pr(\cP_{b,g})+\Pr(\cU_{b,g})+\Pr(\cS_{b,g}). \label{eq:outage1}
\eeq
Substituting  the physical blockage probability $p_{b,g}^{\text{blk}}$ in \eqref{LoS_Blockage}, UE access-limited blockage $\rho_{b, g}(\bx_b)$ in \eqref{express1}, and SINR outage  upper bound
$\Pr\left(\overline{\text{SINR}}_{b,g}(\mathbf{y},\mathbf{X})\!<\! z| \mathbf{y},\mathbf{X}\!\right)$  in 
\eqref{eqn:pUB} in \eqref{eq:outage1}, an upper bound of the single-link outage probability  $p_{b,g}^{\text{Out}}$ is given by 
\beq
p_{b,g}^{\text{Out}} \!\leq\!
p^{\text{blk}}_{b,g} \!+\! \rho_{b, g}(\bx_b)\!\left(1\!-\!p^{\text{blk}}_{b,g}\right)\!+\!  \left(1\!-\!\rho_{b, g}(\bx_b)\right) \widehat{p}_{b,g}^{\text{SINR}}(\mathbf{y}, \mathbf{X}), 
\label{eq:outage_upper_bound}
\eeq
where $\widehat{p}_{b,g}^{\text{SINR}}(\mathbf{y}, \mathbf{X})\triangleq\left(1\!-\!p^{\text{blk}}_{b,g}\right)\Pr\left(\overline{\text{SINR}}_{b,g}(\mathbf{y},\mathbf{X})\!<\! z| \mathbf{y},\mathbf{X}\!\right)$ for ease of exposition.}

\subsubsection{UE Outage Constraint}
Since a UE $g$ can be covered by multiple BSs, the UE outage in the network occurs when all these links are simultaneously in the outage.
\tcg{Assuming independent outage per link, the UE  outage in the network is 
\begin{eqnarray} \nonumber
\prod\limits_{b=1}^B \!\left[ p_{b,g}^{\text{Out}} \right]^{x_{b,g}} \leq \prod\limits_{b=1}^B \!\Big[p^{\text{blk}}_{b,g} \d4\d4&&+ \rho_{b, g}(\bx_b)\!\left(1\!-\!p^{\text{blk}}_{b,g}\right) \!\\
\hspace{-0.3cm}&&+  \left(1\!-\!\rho_{b, g}(\bx_b)\right) \widehat{p}_{b,g}^{\text{SINR}}(\mathbf{y}, \mathbf{X}) \Big]^{x_{b,g}} \!\! \nn,
\end{eqnarray}
where the inequality follows from \eqref{eq:outage_upper_bound}.
Introducing a UE outage tolerance $\zeta\!\in\!(0,1]$, we find a sufficient condition for the UE outage guarantee to satisfy
\beq \nonumber
\prod\limits_{b=1}^B \!\Big[p^{\text{blk}}_{b,g} \!+\! \rho_{b, g}(\bx_b)\!\left(1\!-\!p^{\text{blk}}_{b,g}\right) \!+ \! \left(1\!-\!\rho_{b, g}(\bx_b)\right) \!\!\!\d4\d4&&\widehat{p}_{b,g}^{\text{SINR}}(\mathbf{y}, \mathbf{X}) \Big]^{x_{b,g}} \\
&& \hspace{0.5cm}\leq \zeta.  \label{eq:outage2}
\eeq
Note that any $(\by,\bX)$ guaranteeing  \eqref{eq:outage2} ensures
$\prod_{b=1}^B \!\big[ p_{b,g}^{\text{Out}} \big]^{x_{b,g}}\leq \zeta$.
For ease of manipulation, it is customary to convert the geometric terms in \eqref{eq:outage2} to a linear form by taking the logarithm on both sides of \eqref{eq:outage2}, which yields}
\begin{eqnarray}  \nonumber
\hspace{-0.3cm} \sum_{b\in \cB} x_{b,g}  \log \Big( \d4\d4\d4&& p^{\text{blk}}_{b,g} + \rho_{b, g}(\bx_b)\!\left(1-p^{\text{blk}}_{b,g}\right) \\
\hspace{-0.3cm} && \hspace{0.5cm} + \left(1\!-\!\rho_{b, g}(\bx_b)\right) \widehat{p}_{b,g}^{\text{SINR}}(\mathbf{y}, \mathbf{X})  \Big) \leq \log\left(\zeta\right). \quad\;\;
\label{SINR_constraint}
\end{eqnarray}

\tcg{The UE outage tolerance level $\zeta$ in \eqref{SINR_constraint} is a user-defined variable and could  be determined practically during the BS deployment planning. Different grid $g$ can be set with different $\zeta$ values in \eqref{SINR_constraint}.
Imposing \eqref{SINR_constraint} as a constraint for the BS deployment, the obtained BS deployment will guarantee successful channel access with a probability larger than $1-\zeta$. In the urban environments with dense UE distribution, the $\zeta$ values per grid can be enforced to be small for reliable channel access, while relatively large $\zeta$ values can be used in rural scenarios.
When the grid has no UEs, we can set $\zeta=1$ parsimoniously, implying that this grid does not need to be covered by any BSs.}

%\subsubsection{UE Outage Constraint Simplification} 
%After incorporating the expressions of $\rho_{b, g}(\bx_b)$ in (\ref{express1}) and $\widehat{p}_{b,g}^{\text{SINR}}(\mathbf{y}, \mathbf{X})$ in (\ref{eq:outage_upper_bound}), the UE outage \eqref{SINR_constraint}  becomes complicate. 
%Particularly, both $\rho_{b, g}(\bx_b)$ and $\widehat{p}_{b,g}^{\text{SINR}}(\mathbf{y}, \mathbf{X})$ are functions of the variable  $\bX$ and have complicate expressions, which makes the later BS deployment problem solving and algorithm design difficult. 
However, the constraint in \eqref{SINR_constraint} is not analyzable and difficult to be used for formulating optimization problems, for which we propose below decomposition of \eqref{SINR_constraint} into two simpler bounds based on the following lemma. 

\begin{lemma}
\label{lemma_UEAccessLimitedBlc}
The left-hand-side of (\ref{SINR_constraint})  
is a monotonically increasing function of the UE access-limited blockage probability  $\rho_{b,g}(\bx_b)$.
\end{lemma}
\begin{proof}
It is not difficult to observe that the left-hand-side of (\ref{SINR_constraint}), which is represented by $\sum_{b\in \cB} x_{b,g}  \log(\Psi_{b,g})$, where
$$\Psi_{b,g}\!\triangleq\!p^{\text{blk}}_{b,g} + \rho_{b, g}(\bx_b)\!\left(1-p^{\text{blk}}_{b,g}\right)  + \left(1\!-\!\rho_{b, g}(\bx_b)\right) \widehat{p}_{b,g}^{\text{SINR}}(\mathbf{y}, \mathbf{X}),$$
is an increasing function of $\Psi_{b,g}$. 
Hence, we only need to validate that the $\Psi_{b,g}$ is a monotonically increasing function of  
$\rho_{b, g}(\bx_b)$. 
It is not difficult to observe that the first-order derivative of $\Psi_{b,g}$ with respect to $\rho_{b, g}(\bx_b)$ is greater than or equal to zero, i.e., 
\beq \nonumber
\frac{\partial \Psi_{b,g}}{\partial \rho_{b, g}(\bx_b)}=\left(1-p^{\text{blk}}_{b,g}\right)-\widehat{p}_{b,g}^{\text{SINR}}(\mathbf{y}, \mathbf{X})\stackrel{(b)}\geq 0, \nn
\eeq
where (b) follows from the definition of $\widehat{p}_{b,g}^{\text{SINR}}(\mathbf{y}, \mathbf{X})$ in  \eqref{eq:outage_upper_bound} and the upper bound
\beq 
\widehat{p}_{b,g}^{\text{SINR}}(\mathbf{y}, \mathbf{X})\d4 &=& \d4 \left(1\!-\!p^{\text{blk}}_{b,g}\right)\Pr\left(\overline{\text{SINR}}_{b,g}(\mathbf{y},\mathbf{X}) z| \mathbf{y},\mathbf{X}\!\right)
\nn \\ 
\d4 & \leq & \d4 
 1\!-\!p^{\text{blk}}_{b,g}, \label{eq:bound1}
 \eeq
 which completes the proof.
\end{proof}
\tcg{We introduce a tolerance level $\gamma$ to limit the value of $\rho_{b,g}(\bx_b)$ in (\ref{SINR_constraint}) as
\beq
\rho_{b, g}(\bx_b) \!=\!\sum_{i=N_{\text{RF}}+1}^{+\infty}  \! \! \frac{E[{n}_b(\bx_b)]^{i}}{i!}e^{-E[{n}_b(\bx_b)]}\frac{i-N_{\text{RF}}}{i} \leq \gamma.
\label{express2}
\eeq
From Lemma~\ref{lemma_UEAccessLimitedBlc} and \eqref{express2}, we further upper bound the left-hand-side of \eqref{SINR_constraint} yielding
\begin{eqnarray}  \nonumber
\hspace{-0.3cm} \sum_{b\in \cB} x_{b,g}  \log \Big( \d4\d4\d4&& p^{\text{blk}}_{b,g} + \gamma\!\left(1-p^{\text{blk}}_{b,g}\right) \\
\hspace{-0.3cm} && \hspace{0.5cm} + \left(1\!-\!\gamma\right) \hat{p}_{b,g}^{\text{SINR}}(\mathbf{y}, \mathbf{X})  \Big) \leq \log\left(\zeta\right). \quad\;\;
\label{SINR_constraint2}
\end{eqnarray}
Since the left-hand-side of (\ref{SINR_constraint2}) is an upper bound of that in  \eqref{SINR_constraint}, any $(\by,\bX)$ guaranteeing  \eqref{SINR_constraint2} ensures \eqref{SINR_constraint}. 
However, the condition in \eqref{express2} is rather difficult to be directly analyzed and used as a constraint of the BS optimization problem.
This motivates us to find an equivalent, but tractable condition of \eqref{express2} to facilitate the BS deployment problem solving. 
From Lemma~\ref{lemma2}, we already know that  $\rho_{b, g}(\bx_b)$ is a monotonically increasing function of $E[{n}_b(\bx_b)]$. Hence, for $\Phi\!\geq\! 0$ satisfying $ \sum_{i=N_{\text{RF}}+1}^{+\infty}  \frac{\Phi^{i}}{i!}e^{-\Phi}\frac{i-N_{\text{RF}}}{i} = \gamma$, the bound in \eqref{express2} is equivalent to the following linear constraint
\begin{eqnarray}
  E[{n}_{b}(\bx_b)] \!=\! \sum_{g\in\mathcal{G}} x_{b,g}\lambda_{\text{UE},g}L^2_{\text{grd}} (1-p_{b,g}^{\text{blk}}) \leq \Phi.
\label{Re_NumUE_PerBS}
\end{eqnarray}
Instead of \eqref{SINR_constraint}, we take in the two inequalities in \eqref{Re_NumUE_PerBS} and \eqref{SINR_constraint2} as the UE outage constraint to formulate the BS deployment optimization problem below.}

%%================================================================================================
%%================================================================================================
\subsection{\tcg{UE Outage-Guaranteed Minimum-Cost BS Deployment Problem}}  \label{Problem_form}
%%================================================================================================
%%================================================================================================

Incorporating the BS association constraints in \eqref{xy_relation}-\eqref{LoS_visible_cons2} and  UE outage constraints  (\ref{SINR_constraint2}), \eqref{Re_NumUE_PerBS} into the minimum-cost BS deployment criterion  gives
\begin{subequations}
\label{Opt_problem}
\begin{eqnarray}
\label{objective}
  \min\limits_{\mathbf{y}, \mathbf{X}}  \d4\d4 &&  \sum_{b=1}^B c_b y_b \\
 \text{subject to} \d4\d4 &&  x_{b,g}\leq  y_{b}, \quad\quad\quad \quad\quad\quad  %\forall b\in\mathcal{B}, g\in \mathcal{G},
                   \label{UE_ass_con} \\
                   &&  x_{b,g} \leq \mathbb{I}_{\left\{{{ r_{b,g}\leq R^{\text{max}}, \;p_{b,g}^{\text{blk}}}<1,  } \right \}}    %\forall b\in\mathcal{B}, g\in \mathcal{G},
                   \label{LoS_Cov} \\
                   &&  \mathbb{I}_{\left\{r_{b,s}\leq r_{b,g}, \;p_{b,s}^{\text{blk}}<1  \right\}}x_{b,g} \leq x_{b,s}, %b\in\mathcal{B}, g\in \mathcal{G},
                   \label{LoS_Cov_Corr} \\
                   && E[{n}_{b}(\bx_b)] \leq \Phi, \;\; \quad  %\forall b\!\in\!\mathcal{B},
                   \label{Cov_Limite} \\ \nonumber
 && \sum_{b\in \cB} x_{b,g} \log\Big(  p^{\text{blk}}_{b,g} \!+\! \gamma\left(1-p^{\text{blk}}_{b,g}\right)  \\
                                 && \hspace{0.3cm} +\left(1\!-\!\gamma\right) \widehat{p}_{b,g}^{\text{ SINR}}(\mathbf{y}, \mathbf{X}) \! \Big) \!\!\leq\! \log\left(\zeta\right),  %\forall g\!\in\!\mathcal{G}, \quad \quad
                   \label{UE_outage} \\  \nonumber
  &&  y_b\!\in\!\{0,1\},  \;x_{b,g}\!\in\!\{0,1\},  \;\; \forall b\in\mathcal{B},\; g\in \mathcal{G},
\end{eqnarray}
\end{subequations}
where $c_b$ in \eqref{objective} is the BS installation cost at location $b\in \cB$.
The objective in \eqref{objective} is to minimize the cost for deploying BSs by jointly optimizing the \emph{BS deployment vector} $\by\in \mathbb{B}^{1\!\times \!B}$  and \emph{association matrix}  $\bX\in \mathbb{B}^{B\!\times \!G}$.
\tcg{Because of the nonlinear constraint  (\ref{UE_outage}) with respect to the binary vector $\by$ and binary matrix $\bX$, the problem  in (\ref{Opt_problem}) is INP, which is excessively complex to be directly solvable \cite{INP1}.} More specifically, directly searching for the optimal solution needs to evaluate all $2^{B\times (G+1)}$ combinations of $(\mathbf{y}, \mathbf{X})$, which is prohibitive for relatively large $B$ and $G$ (e.g., $B\geq 100$ and $G\geq 100$).
%Existing work \cite{mmWaveSiteDeploy} and \cite{MmWaveBsDeploy_UEOre} for mmWave BS deployment suboptimally solve their formulated INP by using heuristic approaches. 
In the next subsection, we address this challenge and propose
%To address this challenge, we propose, in the next section, 
a low-complexity approach to  the problem \eqref{Opt_problem}, while optimally solving it.
\section{Minimum-Cost BS Deployment Algorithm}   \label{sec:two_methods}
%%%================================================================================================
%%%================================================================================================
%%%================================================================================================
%%%================================================================================================
In this section, we  find the optimal solution to the minimum-cost BS deployment problem in \eqref{Opt_problem}. 
The key to optimally solving \eqref{Opt_problem} lies in decomposing it into two \emph{separable} subproblems: (i) BS coverage optimization problem, which  finds a feasible association matrix $\bX$ to the constraints in \eqref{UE_ass_con}-\eqref{UE_outage}  as a function of the BS deployment vector $\by$, and (ii) minimum-cost subset BS selection problem, which finds the minimum-cost $\by$ to guarantee the UE outage constraints.
The main motivation of this approach is that the objective $\sum_{b=1}^B c_b y_b$ in \eqref{Opt_problem} is independent of $\bX$, and thus, the optimal solution can be attained by firstly expressing the feasible $\bX$ (i.e., satisfying  \eqref{UE_ass_con}-\eqref{UE_outage} ) as a function of $\by$ and secondly optimizing $\by$ to minimize the objective function in \eqref{Opt_problem}. 
%The main motivation of this approach is that the cost function $\sum_{b=1}^B c_b y_b$ in \eqref{Opt_problem} is separable with the  association matrix $\bX$, and under the following Lemma~\ref{lemma_MacroDiv_OUTPro}, the optimal solution can be approached by \tcg{first finds the optimal $\bX$ as a function of the BS deployment vector $\by$ under the BS coverage and UE access-limited constraints \eqref{UE_ass_con}-\eqref{Cov_Limite}  }   \tcgr{first solving a feasibility problem that finds the maximum feasible $\bX$ in terms of relaxing the UE outage constraint} and secondly finding the optimal \tcgr{BS deployment vector} $\by$ in terms of minimizing the objective function in \eqref{Opt_problem}. 
The BS coverage optimization subproblem is first discussed.

%%================================================================================================
%%================================================================================================
%%================================================================================================
% ---------------------------System Setup------------------------------------
\subsection{BS Coverage Optimization}  \label{subproblem1}
%%================================================================================================
%%================================================================================================
%%================================================================================================
As a starting point, we introduce the following proposition showing a monotonic relationship between the macro diversity order in \eqref{macro_diversity} and the left-hand-side of the UE outage constraint \eqref{UE_outage}.

%===Another lemma======================================
\begin{lemma}
\label{lemma_MacroDiv_OUTPro}
{For a fixed $\mathbf{y}$, the  left-hand-side of \eqref{UE_outage} is a monotonically decreasing function of the macro diversity order $\sum_{b\in\mathcal{B}} x_{b,g}$} in \eqref{macro_diversity}. 
\end{lemma}
\begin{proof}
We first claim that the left-hanf-side of \eqref{UE_outage} is non-positive, i.e., $$\textstyle \log\Big(p^{\text{blk}}_{b,g} +                \gamma\left(1-p^{\text{blk}}_{b,g}\right)   +\left(1\!-\!\gamma\right) \widehat{p}_{b,g}^{\text{SINR}}(\mathbf{y}, \mathbf{X}) \Big)\leq 0.$$
This can be checked by the bound in \eqref{eq:bound1}. 
%$$\widehat{p}_{b,g}^{\text{SINR}}(\mathbf{y}, \mathbf{X}) \triangleq (1- p_{b,g}^{\text{blk}}) \Pr\left(\overline{\text{SINR}}_{b,g}\!<\! z| \mathbf{y},\mathbf{X}\!\right)\leq (1- p_{b,g}^{\text{blk}}).$$
This reveals that if  $\widehat{p}_{b,g}^{\text{SINR}}(\mathbf{y}, \mathbf{X})$  is a monotonically decreasing  function of  the diversity order $\sum_{b\in\mathcal{B}} x_{b,g}$, so is 
the left-hand-side of \eqref{UE_outage}.
%monotonically decreasing function of the macro diversity order $\sum_{b\in\mathcal{B}} x_{b,g}$ of a grid $g$. 
Thus, in what follows, it suffices to show the monotonicity  of  $\widehat{p}_{b,g}^{\text{SINR}}(\mathbf{y}, \mathbf{X})$. 
To this end, we divide the proof into two cases when $x_{b,g}=1$ and $x_{b,g}=0$.

First, when $x_{b,g}=1$, it can be shown from \eqref{single_inf} and \eqref{SINR_Exp} that as the macro diversity order $\sum_{b\in\mathcal{B}} x_{b,g}$ of a UE in grid $g$ increases for a fixed $\by$, the composite interference power $\sum_{i\in \mathcal{B}} y_i\widehat{I}_{i,g}(x_{i,g})$ decreases, concluding that $\widehat{p}_{b,g}^{\text{ SINR}}(\mathbf{y}, \mathbf{X})$ is a monotonically decreasing  function of   $\sum_{b\in\mathcal{B}} x_{b,g}$. 
On the other hand, when $x_{b,g}=0$, the increment of macro diversity order $\sum_{b\in\mathcal{B}} x_{b,g}$ can lead to either $x_{b,g}=0$ (unchanged) or $x_{b,g}=1$. In the former case, we have
$\widehat{p}_{b,g}^{\text{SINR}}(\mathbf{y}, \mathbf{X})=1-p_{b,g}^{\text{blk}}$ due to \eqref{desirepower}, while in the latter case, $\widehat{p}_{b,g}^{\text{SINR}}(\mathbf{y}, \mathbf{X})$ decreases, i.e.,  $\widehat{p}_{b,g}^{\text{SINR}}(\mathbf{y}, \mathbf{X})\leq 1-p_{b,g}^{\text{blk}}$. As a result, we conclude that $\widehat{p}_{b,g}^{\text{SINR}}(\mathbf{y}, \mathbf{X})$ is a monotonically decreasing function of $\sum_{b\in\mathcal{B}} x_{b,g}$. This  completes the proof.
\end{proof}

This lemma reveals that for a fixed $\by$ the left-hand-side of \eqref{UE_outage} is minimized by maximizing the macro diversity order  $\sum_b x_{b,g}$ of each grid.
Note that for a fixed $\by$ the objective function $\sum_{b=1}^B c_b y_b$ in \eqref{Opt_problem} is
independent of $\bX$ and any $\bX$ that satisfies the constraints in \eqref{UE_ass_con}-\eqref{UE_outage} (i.e., feasible $\bX$) is optimal.
\tcg{By leveraging Lemma \ref{lemma_MacroDiv_OUTPro}, a feasible association matrix $\bX$ of the problem in \eqref{Opt_problem} for a given $\by$  can be obtained by maximizing the macro diversity order  $\sum_b x_{b,g}$  subject to the BS coverage and UE access-limited blockage constraints in \eqref{UE_ass_con}-\eqref{Cov_Limite}. 
%\taejoon{We need to explain here why ``maximizing the macro diversity order  $\sum_b x_{b,g}$ subject to the BS coverage and UE access-limited blockage constraints  \eqref{UE_ass_con}-\eqref{Cov_Limite}" results in the optimal $\bX$ for a given $\by$ in \eqref{Opt_problem}. For example, for a fixed $\by$, can we transform \eqref{Opt_problem} to the problem of `maximizing the macro diversity order  $\sum_b x_{b,g}$ subject to the BS coverage and UE access-limited blockage constraints  \eqref{UE_ass_con}-\eqref{Cov_Limite}"? This should be clearly explained. Otherwise, the reviewers will complain. Though I revised them pretty much, I suggest you checking them throughly}. 
To find this feasible association matrix $\bX$ and express it as a function of a fixed $\by$, 
we introduce an auxiliary variable   $\Lambda_{b,g} \in \B$, called a coverage indicator, associated with the candidate location $b\in\cB$ and the grid $g\in\cG$, such that 
\beq
x_{b,g}= y_b \Lambda_{b,g}, \label{simplex1}
\eeq 
where $\Lambda_{b,g} = 1$ if a candidate BS location $b\in \cB$ (regardless whether $y_b = 1$ or $y_b = 0$) covers the grid $g$, and $\Lambda_{b,g} = 0$ otherwise.
The objective of maximizing the macro diversity order of all grids for a given $\by$
can be expressed and equivalently transformed to the objective of maximizing the cell coverage of each candidate location $b$ as follows, 
\beq
\max_{\bX} \!\sum_{g=1}^G \sum_{b=1}^B x_{b,g} \!=\!\! \max_{\bX}\! \sum_{b=1}^B \sum_{g=1}^G   x_{b,g} \!=\!\! \sum\limits_{b=1}^{B} y_b \max\limits_{\boldsymbol{\Lambda}_b } \!\sum\limits_{g=1}^{G} \Lambda_{b,g}, \label{max_coverage}
\eeq
where the first equality follows from the fact that changing
the order of summations does not alter the optimality and the second equality is due to $x_{b,g} = y_b \Lambda_{b,g}$ and the fact that $\by$ is fixed, where $\sum_{g=1}^{G} \Lambda_{b,g}$ is the cell coverage of the candidate location $b$ and $\boldsymbol{\Lambda}_b = [\Lambda_{b,1}, \ldots, \Lambda_{b,G}]\in \B^{1\times G}$.}
Motivated by \eqref{max_coverage}, we find the \tcg{maximum cell coverage of each candidate location $b$ subject to the BS coverage \eqref{UE_ass_con}-\eqref{LoS_Cov_Corr} and UE  access-limited blockage  \eqref{Cov_Limite}} constraints, leading to
\begin{subequations}
\label{sumproblem1}
\begin{eqnarray}
\max\limits_{\boldsymbol{\Lambda}_b} \d4\d4 && \sum\limits_{g=1}^{G} \Lambda_{b,g}, ~~ \forall b\in \cB \label{Cov_objective} \\
 \text{subject to} \d4\d4 &&  \Lambda_{b,g} \leq \I_{\{ r_{b,g}\leq R^{\text{max}}, p_{b,g}^{\text{blk}} <1 \} }
 \label{LoS_Cov1} \\
 \d4\d4 &&  \mathbb{I}_{\left\{r_{b,s}\leq r_{b,g}, \;p_{b,s}^{\text{blk}}<1  \right\}}\Lambda_{b,g} \leq \Lambda_{b,s},
 \label{LoS_Cov_Corr1} \\
 \d4\d4 && E[{n}_{b}(\boldsymbol{\Lambda}_b)] \leq \Phi, \;\; \quad
 \label{Cap_Limit1} \\
 \nonumber
 \d4\d4 && \Lambda_{b,g}\!\in\!\{0,1\},  \;\; \forall \; g\in \mathcal{G},
\end{eqnarray}
\end{subequations}
where
\beq
E[{n}_{b}(\boldsymbol{\Lambda}_b)] \= \sum_{g\in\mathcal{G}} \Lambda_{b,g}\lambda_{\text{UE},g}L^2_{\text{grd}} (1-p_{b,g}^{\text{blk}}),
\label{LambaE}
\eeq
and the association constraint in \eqref{UE_ass_con} is omitted because the construction in \eqref{simplex1} already implies \eqref{UE_ass_con}. 

We note that all constraints in \eqref{sumproblem1} are consistent with those in \eqref{Opt_problem} except for that it excludes the UE outage constraint in \eqref{UE_outage} and $x_{b,g}$ is changed to $\Lambda_{b,g}$.
Without loss of optimality, we relegate  $E[{n}_{b}(\bx_b)]\leq \bPhi$ in \eqref{Cov_Limite} to $E[{n}_{b}(\boldsymbol{\Lambda}_b)] \leq \bPhi$ because $E[{n}_{b}(\boldsymbol{\Lambda}_b)] \leq \bPhi$ implies $E[{n}_{b}(\bx_b)]\leq \bPhi$.

\tcg{Since the coverage maximization at
each candidate location is separable and the grids closer to a candidate location have the priority to be associated with the BS deployed at the location due to \eqref{LoS_Cov_Corr1}, finding  $\max_{\boldsymbol{\Lambda}_b } \sum_{g=1}^{G} \Lambda_{b,g}$, $\forall b\in \cB$ in \eqref{sumproblem1}
 is equivalent to finding the maximum-link distance  $r_b^{\text{max}}$, which can be described as an optimization problem given by
 \begin{eqnarray}
\label{BS_coverage}
r_b^{\text{max}}   \d4 && \d4 = \max\limits_{\boldsymbol{\Lambda}_b} ~\Lambda_{b,g}r_{b,g}\;~ \\ \nonumber
\text{subject to:}\d4\d4 && (\ref{LoS_Cov1}),  \;(\ref{LoS_Cov_Corr1}),\; \eqref{Cap_Limit1},\; r_{b,g}\leq R^{\text{max}}.
\end{eqnarray}
Accounting for the constraints in (\ref{LoS_Cov1}) and   (\ref{LoS_Cov_Corr1}), it is clear that as the BS $b$ covers more grids (i.e., $\sum_{g=1}^{G} \Lambda_{b,g}$ increases) the objective in \eqref{BS_coverage} increases. 
From \eqref{LambaE} it is also clear that $E[{n}_{b}(\boldsymbol{\Lambda}_b)]$ in  \eqref{Cap_Limit1} is a monotonically increasing function of  $\sum_{g=1}^{G} \Lambda_{b,g}$.  
Hence the $r_b^{\text{max}}$ in \eqref{BS_coverage} is attained either when (i) $E[{n}_{b}(\boldsymbol{\Lambda}_b)] = \Phi$ and $r_b^{\text{max}} < R^{\text{max}}$ or (ii) $E[{n}_{b}(\boldsymbol{\Lambda}_b)] \!\leq\! \Phi$ and $r_b^{\text{max}}\!=\!R^{\text{max}}$.}

Due to monotonicity of the constraint \eqref{Cap_Limit1} with respect to $r_b^{\text{max}}$ (i.e., $ \sum_{g=1}^{G} \Lambda_{b,g}$), searching $r_b^{\text{max}}$ in \eqref{BS_coverage} can be efficiently done by iterative feasibility testing.
Based on these facts, a bisection method for solving \eqref{BS_coverage} is presented as Algorithm~\ref{alg:R_search}.
At each iteration with a given $r_b^{\text{max}}$, Algorithm~\ref{alg:R_search} exploits the characteristics that the constraints  \eqref{LoS_Cov1} and  \eqref{LoS_Cov_Corr1} determine which grids are associated with the candidate location $b$, while the  UE  access-limited constraint \eqref{Cap_Limit1} and $R^{\text{max}}$  examine the feasibility of the $r_b^{\text{max}}$.
The Step \ref{alg:step5} of Algorithm~\ref{alg:R_search} identifies the indicator $\Lam_{b,g}$ such that $\Lam_{b,g}=1$  if $r_{b,g} \leq r_b^{\text{max}}$ and $p_{b,g}^{\text{blk}} < 1$, and  $\Lam_{b,g}=0$ otherwise, $\forall b \in \cB$, $\forall g\in \cG$. It also exploits the constraint \eqref{LoS_Cov_Corr1}; provided  $\Lam_{b,g} =1$, for other grid  $s\in \cG$ we have $\Lam_{b,s} =1$ if 
$r_{b,s} \leq r_{b,g}$ and $p_{b,s}^{\text{blk}} < 1$, and $\Lam_{b,s}=0$ otherwise.
Algorithm~\ref{alg:R_search} requires exactly $\lceil \log_2(R^{\text{max}}/\epsilon) \rceil$ iterations.

%-----Algorithm 1---------------------------------
\begin{algorithm}
\caption{Solving max coverage problem in \eqref{BS_coverage}, $\forall b\!\in\!\cB$}
\label{alg:R_search}
\begin{algorithmic}[1]
\State \textbf{Initialize} Lower bound $LB\!=\!0$, upper bound $UB\!=\!R^{\text{max}}$, middle point $MD\!=\!0$, tolerance $\epsilon\!>\!0$,
\For  {$b=1:B$}
\While {$UB-LB > \epsilon$}
\State $MD=\frac{LB+UB}{2}; ~r_b^{\text{max}}\!=\!MD$
\State \label{alg:step5} \tcg{Determine  $\boldsymbol{\Lambda}_b$ using (\ref{LoS_Cov1}), (\ref{LoS_Cov_Corr1})}
    \If{(\ref{Cap_Limit1}) hold }
        \State  {Update $LB=MD$}
    \Else
        \State  {Update $UB=MD$}
    \EndIf
\EndWhile
\EndFor
\State \Return  $r_b^{\text{max}}, \forall b\in\mathcal{B}$.
\end{algorithmic}
\end{algorithm}
%-------------------------------------------------

Once $r_b^{\text{max}}$, $\forall b \in \cB$, in \eqref{BS_coverage} are determined, we obtain the optimal coverage indicators $\Lambda_{b,g}^\star, \forall b\in\cB, \forall g\in\cG$, based on Step \ref{alg:step5} of Algorithm~\ref{alg:R_search}.
Then, the association  matrix $\mathbf{X}$ is constructed as a function of the given $\by$ according to \begin{eqnarray}
x_{b,g}^\star =y_b\Lambda_{b,g}^{\star}, \;\; \forall b\in\mathcal{B},\; \forall g\in\mathcal{G}.
\label{simpleX2}
\end{eqnarray}
We denote the optimal $\bX$ obtained in \eqref{simpleX2} as $\bX^\star_{\by}$.
Note that for a fixed $\by$, the optimized $\bX^\star_{\by}$ in \eqref{simpleX2} satisfies the constraints (\ref{UE_ass_con})-(\ref{Cov_Limite}) and minimizes the left-hand-side of \eqref{UE_outage}. 
%The remaining task is to find the $\by$ that minimizes the object in   (\ref{Opt_problem}), and the $(\by,\bX)$ makes the constraint  \eqref{UE_outage} hold.

%%================================================================================================
%%================================================================================================
%%================================================================================================
% ---------------------------System Setup------------------------------------
\subsection{Minimum-Cost  Subset  BS  Selection}  \label{subproblem2}
%%================================================================================================
%%================================================================================================
%%================================================================================================

The remaining task is to find the $\by$ that minimizes the object in   (\ref{Opt_problem}) to guarantee the  UE outage constraint in \eqref{UE_outage}, which leads to the second subproblem:

\begin{subequations}
\label{sumproblem2}
\begin{eqnarray}
\d4\min\limits_{\mathbf{y}}   \d4\d4 && \sum_{b=1}^{B} c_by_b  \\ \nonumber
\d4 \text{subject to} \d4\d4 && \sum_{b\in \cB} x_{b,g} \log\Big(  p^{\text{blk}}_{b,g} \!+\! \gamma\left(1-p^{\text{blk}}_{b,g}\right)  \\ \label{UE_outage1}
                                 && \hspace{0.3cm} +\left(1\!-\!\gamma\right) \widehat{p}_{b,g}^{\text{ SINR}}(\mathbf{y}, \bX^\star_{\by}) \! \Big) \!\!\leq\! \log\left(\zeta\right),  %\forall g\!\in\!\mathcal{G}, \quad \quad
                    \\  \nonumber
  &&  y_b\!\in\!\{0,1\}, \; \forall b\in\mathcal{B}.
\end{eqnarray}
\end{subequations}
%\taejoon{Here, $\bX$ in \eqref{UE_outage1} is an optimized one from the first subproblem and should be presented using a different notation, e.g., $\bX^{\star}_{\by}$ or $\bX^{\star}(\by)$}
Using \eqref{SINR_Exp}, the $\widehat{p}_{b,g}^{\text{ SINR}}(\mathbf{y}, \bX^\star_{\by})$ in \eqref{UE_outage1} can be rewritten as
\begin{eqnarray}   \nonumber
\widehat{p}_{b,g}^{\text{ SINR}}(\mathbf{y}, \bX^\star_{\by}) \d4\d4\!\! &&\stackrel{(a)}=   (1\!-\!p_{b,g}^{\text{blk}}) \I_{\left\{ \frac{\overline{P}_{b,g}(y_b\Lambda_{b,g}^{\star}) }{  \sigma^2 + \sum\limits_{i\in \mathcal{B}} y_i \widehat{I}_{i,g}(y_{i}\Lambda_{i,g}^{\star}) } < z \right\}} \quad \\
&&\stackrel{(b)}= \textstyle (1\!-\!p_{b,g}^{\text{blk}})  \I_{\left\{{\frac{ y_b \overline{P}_{b,g}(\Lambda_{b,g}^{\star}) }{  \sigma^2 + \sum\limits_{i\in \mathcal{B}} y_i \widehat{I}_{i,g}(\Lambda_{i,g}^{\star})   } < z} \right\}}
\label{Upperbound3}
\end{eqnarray}
where (a) is because of \eqref{simpleX2}, and (b) is due to the fact that  both  $\overline{P}_{b,g}(y_b\Lambda_{b,g}^\star)$ and $y_b \widehat{I}_{b,g}(\Lambda_{b,g}^\star)$  are equal to $0$ if $y_b=0$.
By doing so, we manipulate the  UE outage constraint in \eqref{UE_outage1} so that it only depends on $\by$.
The problem in \eqref{sumproblem2} is therefore reformulated as
\begin{subequations}
\label{Opt_subproblem2}
\begin{eqnarray}
   \min\limits_{\mathbf{y}} \d4\d4 &&  \sum_{b=1}^B c_by_b  \\ \nonumber
 \text{subject to}\;\d4\d4&&  \sum_{b=1}^B y_b\Lambda_{b,g}^\star  \log\bigg(  p^{\text{blk}}_{b,g} \!+\! \gamma\left(1-p^{\text{blk}}_{b,g}\right) \\
 \label{SINR_constraint_updated2}
                                         && + \left(1\!-\!\gamma\right) \widehat{p}_{b,g}^{\text{ SINR}}(\mathbf{y})  \bigg)\leq \log\left(\zeta\right),  \\ \nonumber
                                 \label{Domain}
                                         &&  y_b\in\{0,1\},  \forall b\in\mathcal{B},
\end{eqnarray}
\end{subequations}
which is INP because of the nonlinear constraint in  (\ref{SINR_constraint_updated2}). 
\tcg{As aforementioned, finding an optimal solution of large-scale INP  (e.g., large $B$)  is often impractical \cite{sub6_siteBS1, sub6_siteBS2,sub6_siteBS3,sub6_siteBS4,sub6_siteBS5,BSSleep_Femto, lagrangian_dual}.}
Rather than proposing another suboptimal treatment of INP, we propose a sequence of linearlization procedures in the following lemma to equivalently transform the constraint in  (\ref{SINR_constraint_updated2}) to a set of linear constraints so that the problem in \eqref{Opt_subproblem2} can be  converted to integer linear programming (ILP).

%\taejoon{change all optimized $\Lam_{b,g}$ to $\Lam_{b,g}^\star$}

\begin{lemma}
\label{Const_linearization}
Suppose auxiliary variables  $s_{b,g}\!\in\!\B$, $\forall b\!\in\!\mathcal{B},~\forall g\!\in\!\mathcal{G}$, which is determined by the indicator function in \eqref{Upperbound3}:
\begin{eqnarray}
\I_{\left\{{\frac{ y_b \overline{P}_{b,g}(\Lambda_{b,g}^\star) }{  \sigma^2 + \sum\limits_{i\in \mathcal{B}} y_i \widehat{I}_{i,g}(\Lambda_{i,g}^\star)   } < z} \right\}}=1-s_{b,g}.
\label{s_notation}
\end{eqnarray} Then, the nonlinear  constraint in  (\ref{SINR_constraint_updated2}) is equivalent to the following set of linear constraints,
\begin{subequations}
\begin{eqnarray}
\d4\d4\d4\d4\d4&& s_{b,g}\leq y_b\Lambda_{b,g}^\star, \;\;  s_{b,g}\in\B, \label{st_domain_0} \\
\d4\d4\d4\d4\d4&& \sigma^2 \+ { \sum\limits_{i\in \mathcal{B}} y_i \widehat{I}_{i,g}(\Lambda_{b,g}^\star) } \geq \frac{y_b\overline{P}_{b,g}(\Lambda_{b,g}^\star)}{z}-s_{b,g}M_{b,g}, \label{LB_con2_0}
\\
\d4\d4\d4\d4\d4&& \sigma^2\+\sum\limits_{i\in \mathcal{B}} y_i \widehat{I}_{i,g}(\Lambda_{i,g}^\star) <  \frac{y_{b}\overline{P}_{b,g}(\Lambda_{b,g}^\star)}{z}\+\left(1\! -\!s_{b,g}\right)M_{b,g}, \label{UB_con2_0}
\\
\d4\d4\d4\d4\d4&&  \sum_{b\in\mathcal{B}} s_{b,g}\log\!\left( p^{\text{blk}}_{b,g} \!+\! \gamma(1-p^{\text{blk}}_{b,g})  \right)  \leq  \log\left(\zeta\right), \label{SINR_Cons3_nonlinear_0}
\end{eqnarray}
\end{subequations}
where $z$ is the link SINR threshold in \eqref{Upperbound3} and    $M_{b,g}=2\sigma^2+ {\sum_{i\in \mathcal{B}} \widehat{I}_{i,g}(\Lambda_{i,g}^\star)}$.
\end{lemma}

\begin{proof}
It is not difficult to observe that if $y_b\Lambda_{b,g}^\star =1$, the $s_{b,g}$ in \eqref{s_notation} is either $s_{b,g}=0$ or $s_{b,g}=1$, while if $y_b\Lambda_{b,g}^\star =0$, then $s_{b,g}=0$, leading to
\begin{eqnarray}
\label{st_domain}
s_{b,g} \leq y_b\Lambda_{b,g}^\star, \;\;    s_{b,g}\in\{0,1\}.
\end{eqnarray}
Because $M_{b,g}>\sigma^2 + { \sum_{i\in \mathcal{B}} y_i \widehat{I}_{i,g}(\Lambda_{i,g}^\star) }, \forall \by$, the indicator function in \eqref{s_notation} can be  equivalently expressed as the following two linear equations:
\begin{eqnarray}
\label{UB_con1}
\d4\d4\d4\d4 && \sigma^2+ { \sum\limits_{i\in \mathcal{B}} y_i \widehat{I}_{i,g}(\Lambda_{b,g}^\star) } \geq \frac{y_b\overline{P}_{b,g}(\Lambda_{b,g}^\star)}{z}-s_{b,g}M_{b,g}  \\ \nonumber
\d4\d4 \text{and} \d4\d4\d4\d4
 \\
\label{UB_con2}
\d4\d4\d4\d4 && \sigma^2 + {  \sum\limits_{i\in \mathcal{B}} y_i \widehat{I}_{i,g}(\Lambda_{b,g}^\star) } < \frac{y_b\overline{P}_{b,g}(\Lambda_{b,g}^\star)}{z}\!+\!(1\!-\!s_{b,g})M_{b,g}.\;\;
\end{eqnarray}

For the $s_{b,g}$ satisfying \eqref{st_domain}-\eqref{UB_con2}, the $\widehat{p}_{b,g}^{\text{ SINR}}$ in \eqref{Upperbound3} can be simplified to
\begin{eqnarray}
\widehat{p}_{b,g}^{\text{ SINR}}=(1-p_{b,g}^{\text{blk}}) \left(1\!-\!s_{b,g} \right).
\label{p_bg}
\end{eqnarray}
Plugging \eqref{p_bg} in the logarithm term
$\log (  p^{\text{blk}}_{b,g} + \gamma (1-p^{\text{blk}}_{b,g})  + (1 - \gamma) \widehat{p}_{b,g}^{\text{ SINR}}(\mathbf{y}))$ on the left-hand-side of (\ref{SINR_constraint_updated2}) and incorporating the two cases,  $s_{b,g}=0$ and $s_{b,g}=1$, into it lead to
\beq
s_{b,g}\log\!\left( p^{\text{blk}}_{b,g} \!+\! \gamma(1-p^{\text{blk}}_{b,g})  \right). \label{log_tmp}
\eeq
Therefore the UE outage constraint in  (\ref{SINR_constraint_updated2})
is succinctly
\begin{eqnarray} \nonumber
 \sum_{b\in\mathcal{B}} y_b\Lambda_{b,g}^\star s_{b,g}\log\!\left( p^{\text{blk}}_{b,g} \!+\! \gamma(1-p^{\text{blk}}_{b,g})  \right)  \!\leq\!  \log\left(\zeta\right), \; \forall g\!\in\!\mathcal{G},
\label{SINR_Cons3_nonlinear2}
\end{eqnarray}
which is still nonlinear with respect to the variables $y_b$ and $s_{b,g}$ because they are coupled.
However, $s_{b,g}\leq y_{b}\Lambda_{b,g}^\star$ in (\ref{st_domain}) implies
\begin{equation}
s_{b,g} \!=\! y_b\Lambda_{b,g}^\star s_{b,g},
\end{equation}
which is obtained by multiplying $s_{b,g}$ to the both sides of (\ref{st_domain}).
The linearized UE outage constraint is then given by
\begin{eqnarray}
 \sum_{b\in\mathcal{B}} s_{b,g}\log\!\left( p^{\text{blk}}_{b,g} \!+\! \gamma(1-p^{\text{blk}}_{b,g})  \right) \leq  \log\left(\zeta\right), \; \forall g\!\in\!\mathcal{G}.
\label{SINR_Cons3_nonlinear3}
\end{eqnarray}
In summary, the nonlinear constraint  (\ref{SINR_constraint_updated2})  can be replaced by the linear constraints (\ref{st_domain})-\eqref{UB_con2}  and (\ref{SINR_Cons3_nonlinear3}), which completes the proof.
\end{proof}

Lemma~\ref{Const_linearization} allows us to transform the INP in \eqref{Opt_subproblem2} to ILP:
\begin{eqnarray}
\label{Opt_subproblem2_1}
   \min\limits_{\mathbf{y}}  \sum_{b=1}^B c_by_b  
 \quad \text{subject to}\;\eqref{st_domain_0}, \eqref{LB_con2_0}, \eqref{UB_con2_0}, \eqref{SINR_Cons3_nonlinear_0}.
\end{eqnarray}

\subsection{Overall Algorithm}
\tcg{We now present our overall framework for finding the optimal solution to the minimum-cost BS deployment problem in \eqref{Opt_problem}. The optimal  $\{ {\Lambda}_{b,g}^\star \}$, $\forall b\in \cB, ~ \forall g\in \cG$, are obtained by Algorithm \ref{alg:R_search} that solves the problem in \eqref{BS_coverage} (equivalently, \eqref{sumproblem1}). After attaining the optimal association matrix $\bX^\star_{\by}$ as a function of the BS deployment vector $\by$}, the problem in \eqref{Opt_subproblem2_1} is solved to obtain the minimum-cost BS deployment $\by$. We notice that the ILP in \eqref{Opt_subproblem2_1} is a standard integer programming, which can be globally solved by the branch-and-bound (B$\&$B) method \cite{INP1}. 
Since it is a standard procedure and there are numerous efficient solvers (e.g., Gurobi \cite{Gurobi}), we omit the details here.

\subsubsection{Nulling Variables for the Computational Complexity Reduction}
\label{newVar_lowComp}
A drawback of the linearization in \eqref{Opt_subproblem2_1} (respectively, Lamma \ref{Const_linearization}) is
that  binary auxiliary variables $s_{b,g},\forall b\in\mathcal{B}, g\in\mathcal{G}$ are additionally introduced, which will stymie the computation of the  B\&B method.
Nevertheless, according to the fact that $s_{b,g}\leq y_b\Lambda_{b,g}$ in \eqref{st_domain}, we can reduce the number of variables by setting
$s_{b,g}=0$ when $\Lambda_{b,g}^\star=0$. Moreover, in the B\&B method, one can effectively reduce the number of branches; when an element $y_b$ in $\by$ is branched into $y_b=0$, values of the auxiliary variables $s_{b,g}, \forall b$ becomes zero. In this way, the increased computational complexity due to the introduced $\{s_{b,g}\}$ is reasonably reduced. Simulation results in the next section will corroborate these conclusions.

\begin{figure*}%
	\centering
	\subfloat[3D campus of the {University of Kansas}.]{{\includegraphics[width=6.5cm]{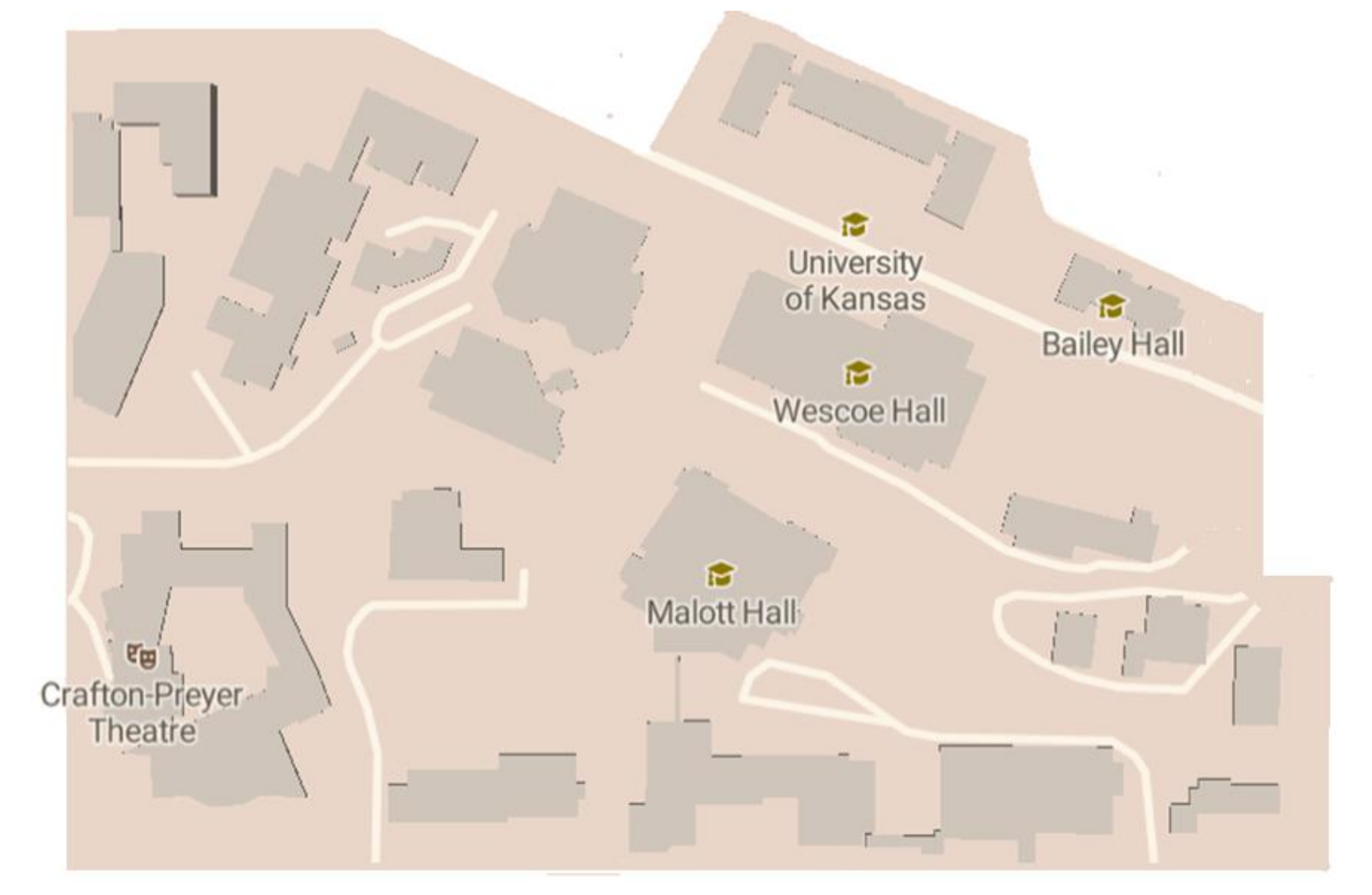} }}%
	\qquad
	\subfloat[Illustration of campus with different number $B$ of candidate BS locations (i.e., $B\!=\!240$ (blue circles), $ B\!=\!184$ (blue circles+red$\ast$), and  $ B\!=\!130$ (blue circles+red$\ast$+green triangle))  on the building walls and the square grids partitioning the outdoor campus. Five regions with the different active UE densities: the UE density at the $i$th region is   $\lambda_{\text{UE},g}^{(i)}\!=\!(2i+2)\times 10^{-4}$, $i=1, 2,\ldots, 5$.]{{\includegraphics[width=10cm]{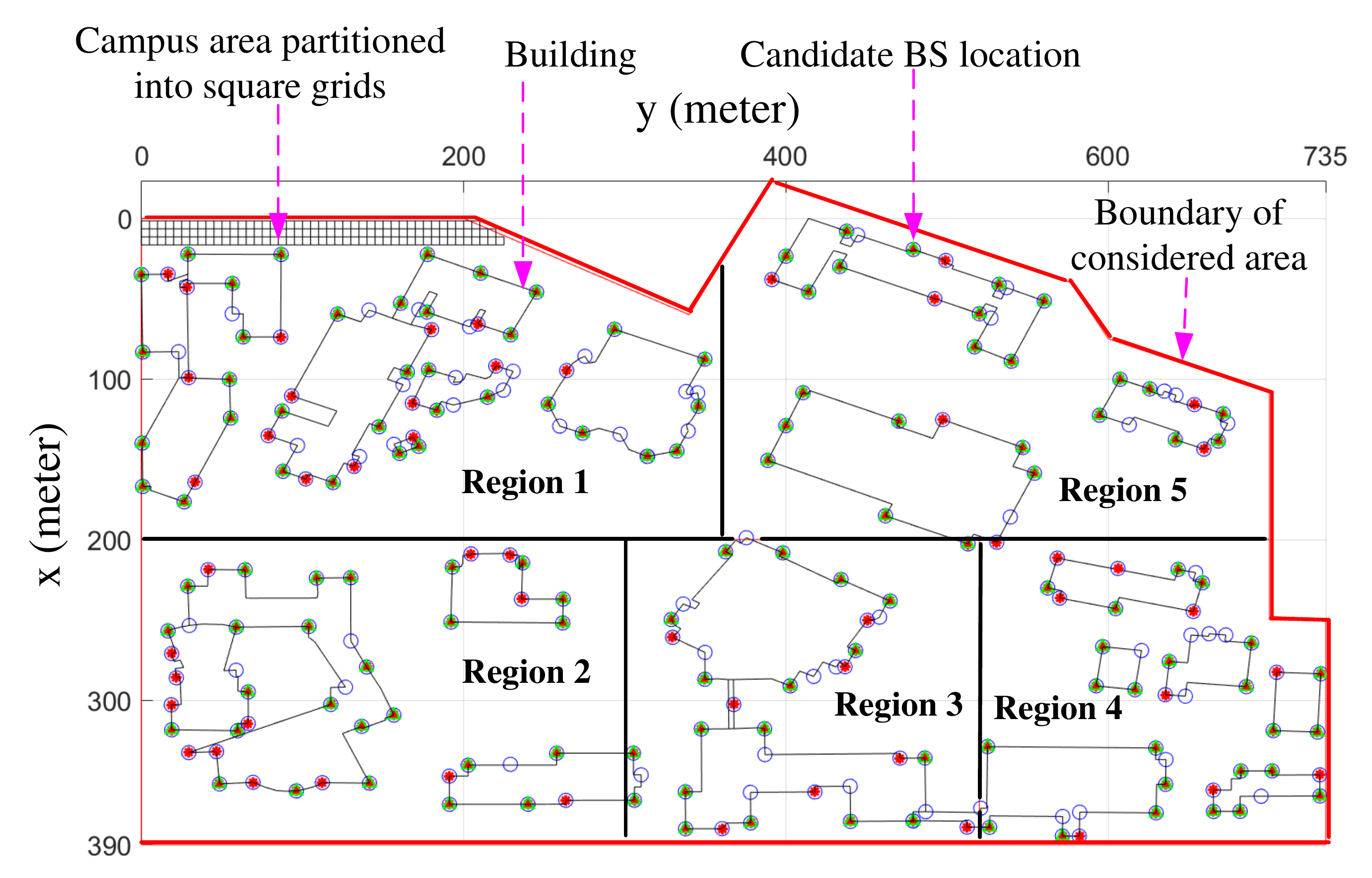} }}%
	\caption{Campus map of the University of Kansas for outdoor mmWave BS deployment evaluation.}%
	\label{fig:geometry2}
\end{figure*}

%\begin{figure} [t]
%  \centering
%  \includegraphics[width=0.5\textwidth]{graph/geometryGroup.pdf}
%  \caption{Five regions with the different active UE densities: the UE density at the $i$th region is   $\lambda_{\text{UE},g}^{(i)}\!=\!(2i+4)\times 10^{-4}$, $i=1, 2,\ldots, 5$.}%
%  \label{fig:geometryPartition}
%\end{figure}
\begin{table}[t]
\scriptsize %\footnotesize  %\small% or footnotesize, scriptsize, tiny, etc.
  \caption{Simulation Parameters }
  \centering
  \begin{tabular}{l|l}
\hline

\textbf{Variables and Description}&
\textbf{Values} \\
\hline
\hline
BS height &
$H_{\text{BS}}=10$ m\\
\hline
UE height  $H_{\text{UE}}$ &
$H_{\text{UE}}=1.5$ m\\
\hline
Square grid length $L_{\text{grd}}$&
$L_{\text{grd}}=5$ m\\
\hline
Numbers of candidate BS locations &
$B=240/184/130$\\
\hline
Numbers of grids &
$G=7393$\\
\hline
 Parameters for physical blockage in \eqref{LoS_Blockage}  &
$\alpha=0.08$, $\beta=0.08$
\\
\hline
Number of RF chains  $N_{\text{RF}}$&
 $N_{\text{RF}} = 12 $ \\
\hline
Maximum link distance $R^{\text{max}}$&
 $R^{\text{max}} = 200$ m \cite{Link_distance}\\
\hline
Link SINR threshold $z$ &
$z\!=\!1$ \\ %(\tcr{Specify its unit, W or dBm?}) \\
\hline
UE access-limited blockage tolerance
& $\gamma=0.05$ \\
\hline
UE outage tolerance &
$\zeta = 0.05$\\
\hline
Mainlobe  and sidelobe  beam gain  &
$G_{\text{main}}\!=\!15$ dB and $G_{\text{side}} \!=\!-9$ dB \cite{BeamPattern} \\
\hline
BS transmit power&
$P_\text{Tx} = 1$ Watt \cite{BsTxPower}\\
\hline
Noise power $\sigma^2$&
$\sigma^2 = -104.5$ dBm \cite{NoisePower}\\
\hline
Tolerance $\epsilon$ in Algorithm~1&
$\epsilon=0.1$\\
\hline
\end{tabular}
\label{tab:1}
\end{table}

%\begin{figure*}%
%	\centering
%	\subfloat[BS deployment with IDS algorithm.] {{\includegraphics[width=8.5cm]{graph/IDS_BSDeployment.pdf} }}%
%	\qquad
%	\subfloat[BS deployment with B\&B algorithm.]    {{\includegraphics[width=8.5cm]{graph/Opt_BSDeployment.pdf} }}%
%	\caption{Proposed mmWave BS deployment  in urban street geometry.}%
%	\label{fig:BSDeployment_proposed}
%\end{figure*}

%%================================================================================================
%%================================================================================================
%%================================================================================================
%%================================================================================================
\section{Simulation Studies}     \label{Sec_Simulation}
%%================================================================================================
%%================================================================================================
%%================================================================================================
%%================================================================================================
\tcg{In this section, we numerically evaluate the proposed BS deployment scheme  in terms of the deployment cost, computational complexity, and UE outage performance. 
The geometry in Fig.~\ref{fig:geometry2} with  dimension $390$ m $\times$ 735~m is considered to evaluate the performance of the proposed BS deployment scheme. Different numbers of candidate BS locations (i.e., $B\!=\!240, B\!=\!184, B=\!130\!$ in Fig.~\ref{fig:geometry2}(b)) are considered to evaluate the tradeoff between the time complexity in solving the transformed minimum-cost BS selection subproblem \eqref{Opt_subproblem2_1} and the UE outage performance. Specific parameters of the geometry in Fig.~\ref{fig:geometry2} and the considered mmWave systems are summarized in   TABLE~\ref{tab:1}.}
Although different grid $g$ could  have different UE density $\lambda_{\text{UE},g}$, we divide the considered geometry  into five distinct regions for simplicity as shown in Fig.~\ref{fig:geometry2} and assume that the grids in the same region have the same UE density, in which the UE density of the $i$th region is described by  $\lambda_{\text{UE}}^{(i)}\!=\!(2i+2)\times 10^{-4}$, $i=1, \ldots, 5$.
Based on the model in Fig. \ref{fig:geometry2}, there are on average $165$ active UEs in the network. Considering the fact that the cost $c_b$ in \eqref{objective} of installing a BS in an area with higher UE density (e.g., urban area) is, in general, more expensive than that of lower UE density (e.g., rural area), we set the installation cost $c_b$ in the $i$th region as $0.2i$, $i=1, \ldots, 5$ for simplicity.

%====================================================================================================
%====================================================================================================
%====================================================================================================
\subsection{\tcg{Benchmark Schemes}}  \label{subsection: paramSetting}
%====================================================================================================
%====================================================================================================
%====================================================================================================
\tcg{We will compare our proposed BS deployment algorithm against the site-specific mmWave BS deployment strategies below. }

\begin{itemize}
   \item \tcg{\emph{Macro Diversity-Constrained Problem  (MDP)}:
The MDP is formed by minimizing the BS deployment cost in \eqref{objective} and by requiring each grid to be covered by at least two BSs: 
    \begin{eqnarray}
\label{Opt_problem:MDP}
 \min\limits_{\mathbf{y}}  \d4\d4 &&  \sum_{b=1}^B c_b y_b \\ \nonumber
 \text{subject to} \d4\d4 &&  \sum_{b=1}^B x_{b,g} \geq 2, \; \forall g\in\cG. 
\end{eqnarray}
The constraint provides a macro diversity guarantee to each grid, which can be effectively used to manage the physical blockage. 
The MDP in \eqref{Opt_problem:MDP} is ILP. 
Thus, it can be efficiently solved by using available solvers. }

  \item \tcg{\emph{Average Signal Strength-Guaranteed Problem (ASSGP) in \cite{mmWaveSiteDeploy}}: The underlying idea is to distribute BSs to guarantee a certain level of average received signal strength (RSS). The ASSGP is therefore formed by adding an additional constraint, setting a  threshold for the RSS of each UE, to the MDP in \eqref{Opt_problem:MDP} below: 
    \begin{eqnarray}
\label{Opt_problem:ASGP}
 \min\limits_{\mathbf{y}}  \d4\d4 &&  \sum_{b=1}^B c_b y_b \\ \nonumber
 \text{subject to} \d4\d4 &&  \sum_{b=1}^B x_{b,g} \geq 2, \; \forall g\in\cG \\ \nonumber
  && \frac{1}{\sum_{b=1}^B x_{b,g}} \sum_{b=1}^B x_{b,g}\text{RSS}_{b,g} \geq \text{RSS}_{\text{th}}, \forall g\in\cG,
\end{eqnarray}
where $\text{RSS}_{b,g}=P_{b,g}+G_{\text{main}}-\text{PL}_{b,g}(r_{b,g})$ in dB is the RSS of the link from BS $b$ to grid $g$ with distance $r_{b,g}$, and the RSS threshold is set to $\text{RSS}_{\text{th}}\!=\!-90$ dB. 
The problem in \eqref{Opt_problem:ASGP} is solved by using the heuristic approach proposed in \cite{mmWaveSiteDeploy}.}

  \item \tcg{\emph{Blockage-Guaranteed Greedy Approach (BGGA)}: 
  In this benchmark, we form a strategy that focuses on providing blockage tolerance. This can be done by replacing the  UE outage constraint (\ref{UE_outage}) of our proposed problem in \eqref{Opt_problem} to \begin{eqnarray}
  \sum_{b\in \cB} x_{b,g} \log\Big(  p^{\text{blk}}_{b,g} \!+\! \gamma\left(1-p^{\text{blk}}_{b,g}\right)   \Big) \!\leq\! \log\left(\zeta\right),\; \forall g\in\cG,
  \label{SimplifiedUeOutage}
  \end{eqnarray}
  which is obtained by removing the SINR outage probability in \eqref{UE_outage}. Similar to the proposed algorithm, we decompose the problem into the two separable subproblems. The BS coverage optimization subproblem is first solved by using the algorithm in  Section-\ref{subproblem1}. To solve the minimum-cost subset BS selection subproblem with the constraint in \eqref{SimplifiedUeOutage}, we adopt the greedy algorithm (GA) proposed in \cite{MmWaveBsDeploy_UEOre}. In the GA, a new BS is added per iteration that guarantees the constraint (\ref{SimplifiedUeOutage}) for the largest number of grids while minimizing the BS deployment cost. The iteration ends when (\ref{SimplifiedUeOutage}) holds for all $G$ grids.}
  
\end{itemize}

%Since the first two benchmarks only addressing the physical blockage for link outage, %there is no limitation for the number of grids associated with a BS.  we can directly set each cell with maximum link distance $r_b^{\text{max}}=R^{\text{max}}$  in \eqref{BS_coverage}. Under the BS association constraints \eqref{xy_relation}, \eqref{LoS_visible_cons}, and \eqref{LoS_visible_cons2}, one can then determine the association matrix $\bX$ as a function of the BS deployment vector $\by$ in the first three benchmarks. \taejoon{the content of this paragraph is not well aligned and difficult to understand. Remove it.}

\begin{figure}
  \centering
  \includegraphics[width=0.45\textwidth]{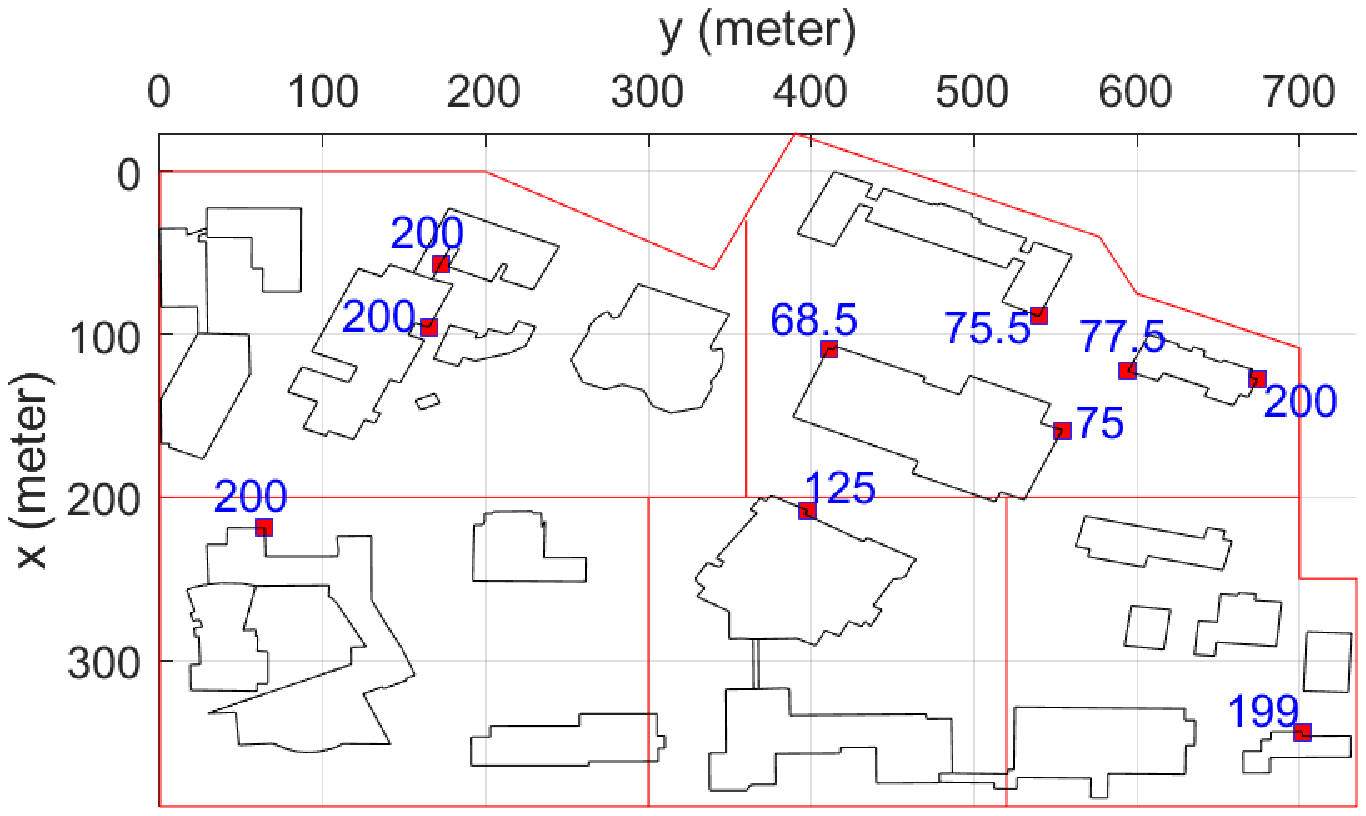}
  %\vspace{-1.5cm}
  \caption{Maximum link distance $r_b^{\text{max}}$ in \eqref{BS_coverage} for $10$  candidate BS locations. The $r_b^{\text{max}}$ values are best represented by the blue fonts.  }
  \label{fig:Prooptimal_CovRadius}
\end{figure}

%===============================================================================================
%===============================================================================================
\subsection{Performance Evaluation}
%===============================================================================================
%================================
\begin{figure*} 
	\centering  %图片全局居中
	%\vspace{-0.35cm} %设置与上面正文的距离
 %设置子图与子标题之间的距离
	\subfloat[MDP scheme,   $61$ deployed BSs with cost $\sum_{b=1}^B c_by_b=32.8$.]    {\includegraphics[width=8.6cm]{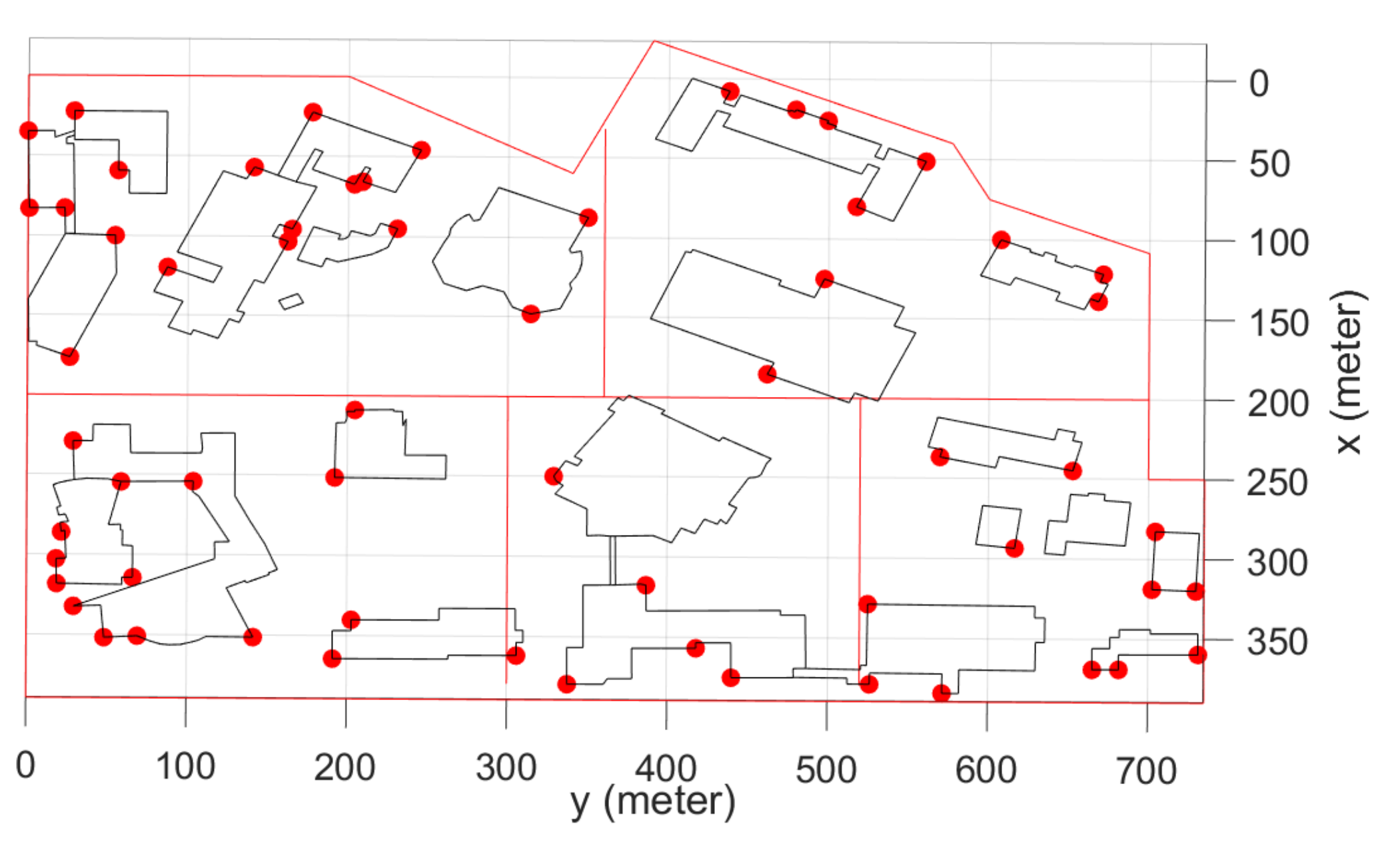}}
	\quad
	\subfloat[ASSGP scheme,   $84$ deployed BSs with cost $\sum_{b=1}^B c_by_b=45$.]  {\includegraphics[width=8.4cm]{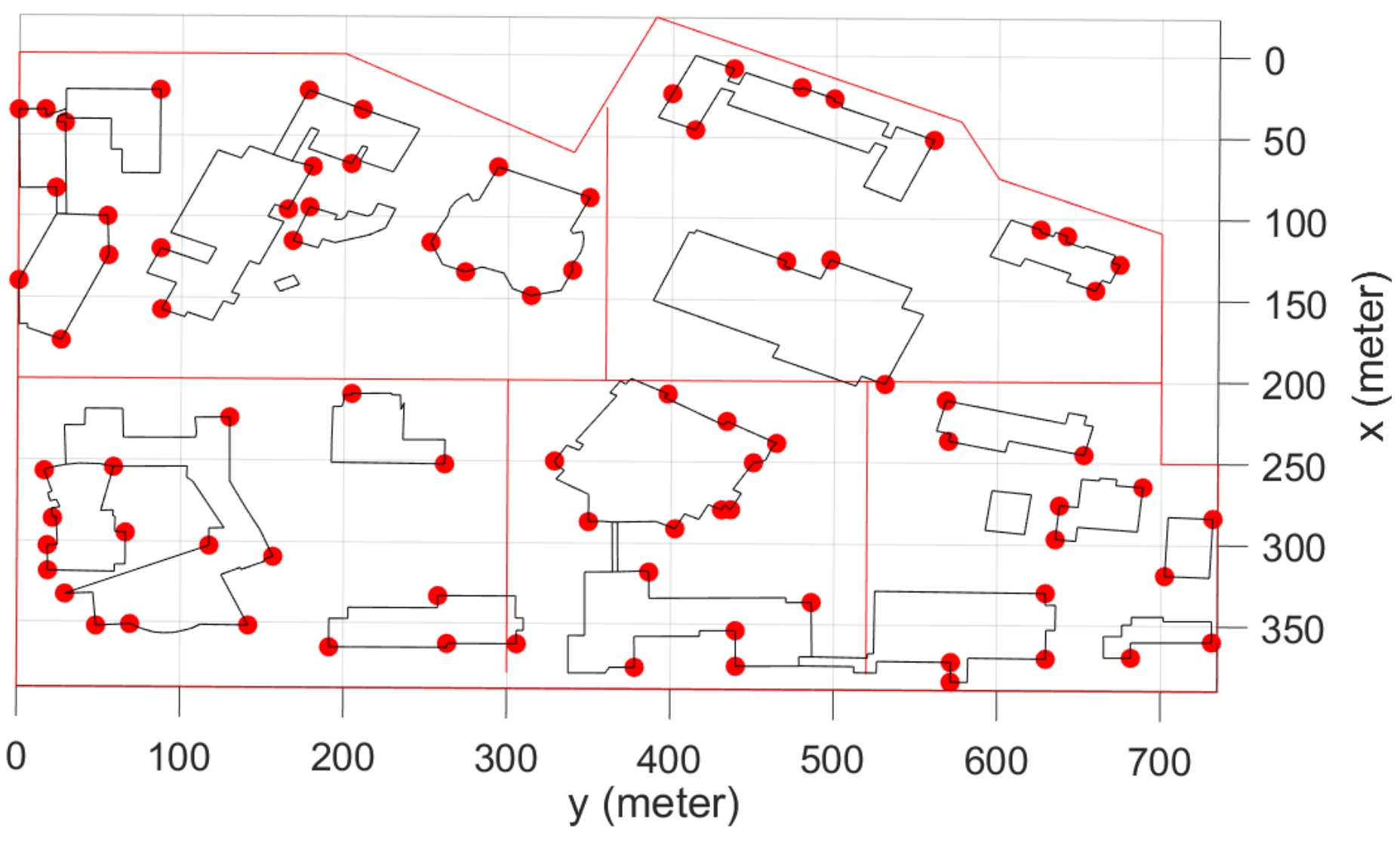}}
	
	  %这里是空了一行，能够实现强制将四张图分成两行两列显示，而不是放不下图了再换行，使用\\也行。
	\subfloat[BGGA scheme,   $102$ deployed BSs with cost $\sum_{b=1}^B c_by_b=55$.]    {{\includegraphics[width=8.6cm]{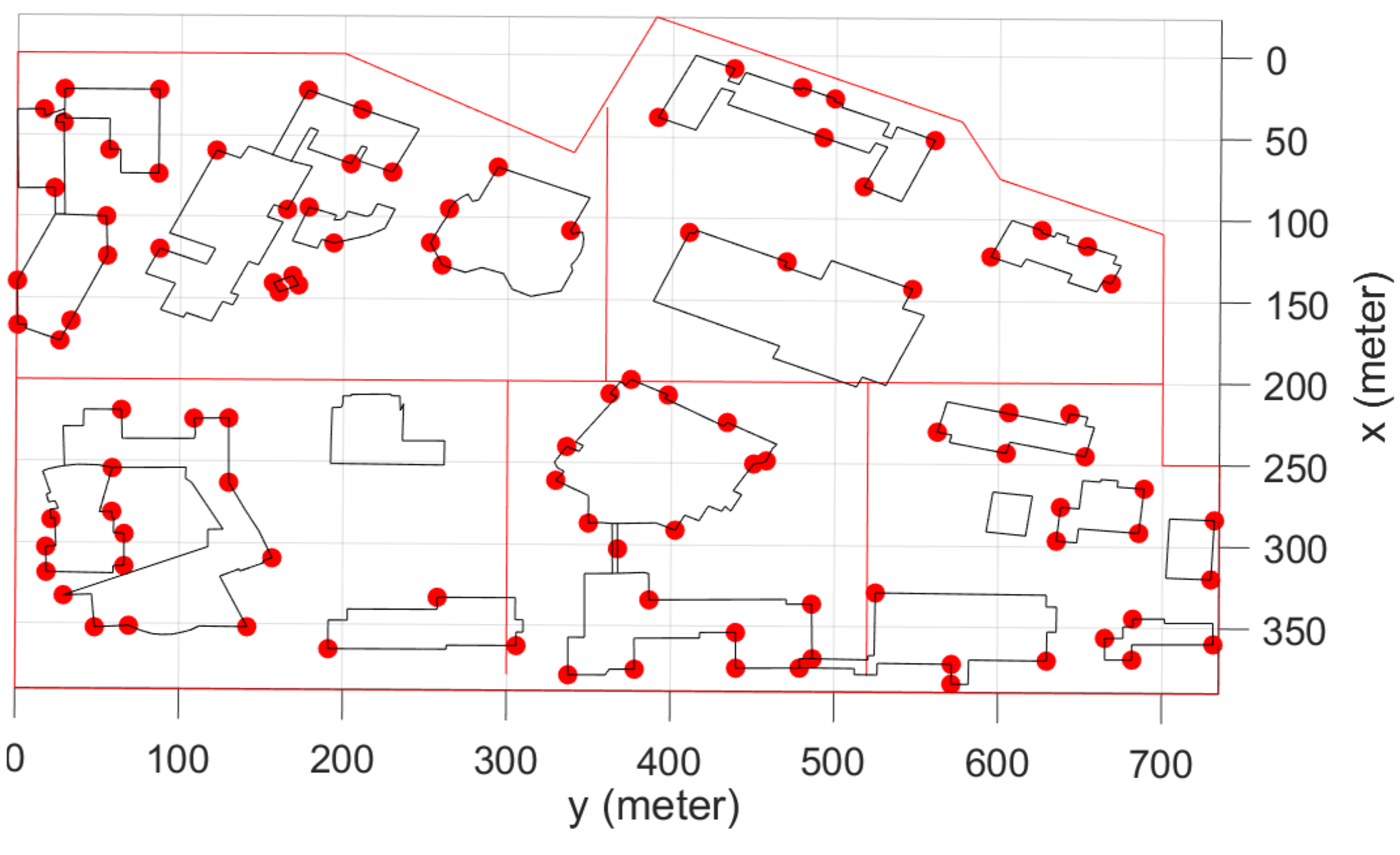}}}
	\quad
	\subfloat[Proposed scheme,   $89$ deployed BSs with cost $\sum_{b=1}^B c_by_b=48.4$.]  {{\includegraphics[width=8.5cm]{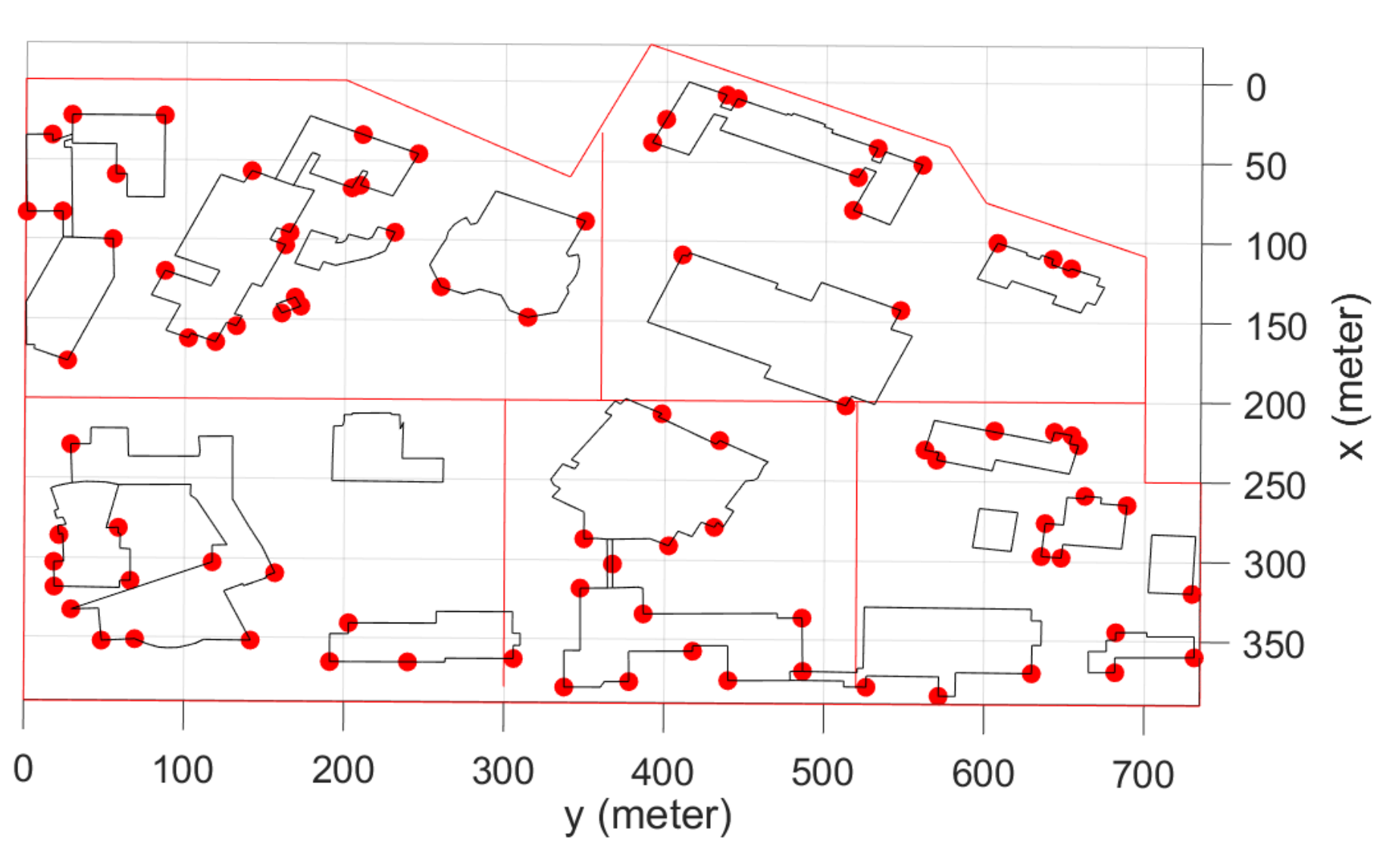}}}
	\caption{BS deployment results of the benchmark MDP, ASSGP, BGGA, and proposed scheme under $B\!=\!240$ candidate locations.}
	\label{BSDeployResults}
\end{figure*}
%===================================

In this subsection, we present the BS deployment results obtained by the proposed scheme and the benchmarks MDP, ASSGP, and BGGA. Using these results, we evaluate and compare the link SINR and UE outage for different schemes. We begin with highlighting the result of the Algorithm~1 that yields the maximum link distance for each candidate BS location.

\subsubsection{BS Coverage Maximization}
Fig.~\ref{fig:Prooptimal_CovRadius} presents the maximum link distance $r_b^{\text{max}}$ values in \eqref{BS_coverage} at the ten different  candidate BS locations, obtained by  Algorithm~\ref{alg:R_search}. 
Because the UEs in Region $5$ have the highest density due to the increased UE access-limited blockage, those candidate locations have relatively small coverage radii\footnote{Those values are $59, 63, 78$, and $200$ meters in Fig.~\ref{fig:Prooptimal_CovRadius}. Since the BS candidate at the boundary of Region $5$ is LoS-visible to a limited number of grids, it has the maximum link distance $200$ meters.} compared to other regions. In contrast, the candidate BS locations in the open areas of Regions $1$ and $2$ are LoS-visible to many grids but most of the maximum link distances are larger than $100$ meters due to the relatively low UE density.
This observation reveals that the maximum link distance depends on both the UE density and the nearby geometry.

\subsubsection{BS Deployment Cost}
Given $B=240$ candidate BS locations shown  as blue circles in Fig.~\ref{fig:geometry2}, the  BS deployment results of the proposed method and the benchmark MDP, ASSGP, and BGGA, are displayed in Fig.~\ref{BSDeployResults}.
The numbers of the deployed BSs of the proposed, MDP, ASSGP, and BGGA schemes are given by $89$, $61$, $84$, and $102$,   respectively.
\tcg{The proposed scheme and BGGA yield larger numbers of deployed BSs due to the UE access-limited blockage constraint, in which a BS $b$ has a  maximum link distance $r_b^{\text{max}}\leq R^{\text{max}}$ as in Fig~\ref{fig:Prooptimal_CovRadius} and can only cover a limited number of grids.
Among the four strategies, the MDP in Fig.~\ref{BSDeployResults}(a) deploys the least number of BSs.   
This happens because the MDP criterion 
merely focuses on extending the LoS link distance to ensure the macro diversity constraint in \eqref{Opt_problem:MDP}. 
In contrast, ASSGA, BGGA, and the proposed scheme attempt to evenly distribute the BSs.
This is because the average RSS constraint   \eqref{Opt_problem:ASGP} in ASSGA, the blockage constraint \eqref{SimplifiedUeOutage} in BGGA, and the UE outage constraint (\ref{UE_outage}) in the proposed scheme control the link distance so that a UE  
 far from its serving BS experiences unsatisfactory link  performance.
 It is not difficult to observe that the BS deployment obtained by the proposed scheme is feasible to the BGGA because the left-hand-side  of \eqref{SimplifiedUeOutage} is a lower bound of that of \eqref{UE_outage} and the BGGA problem is suboptimally solved by the greedy approach in \cite{MmWaveBsDeploy_UEOre},  explaining the inferior performance of BGGA compared to the proposed scheme.}

\begin{figure}
  \centering
  \includegraphics[width=0.5\textwidth]{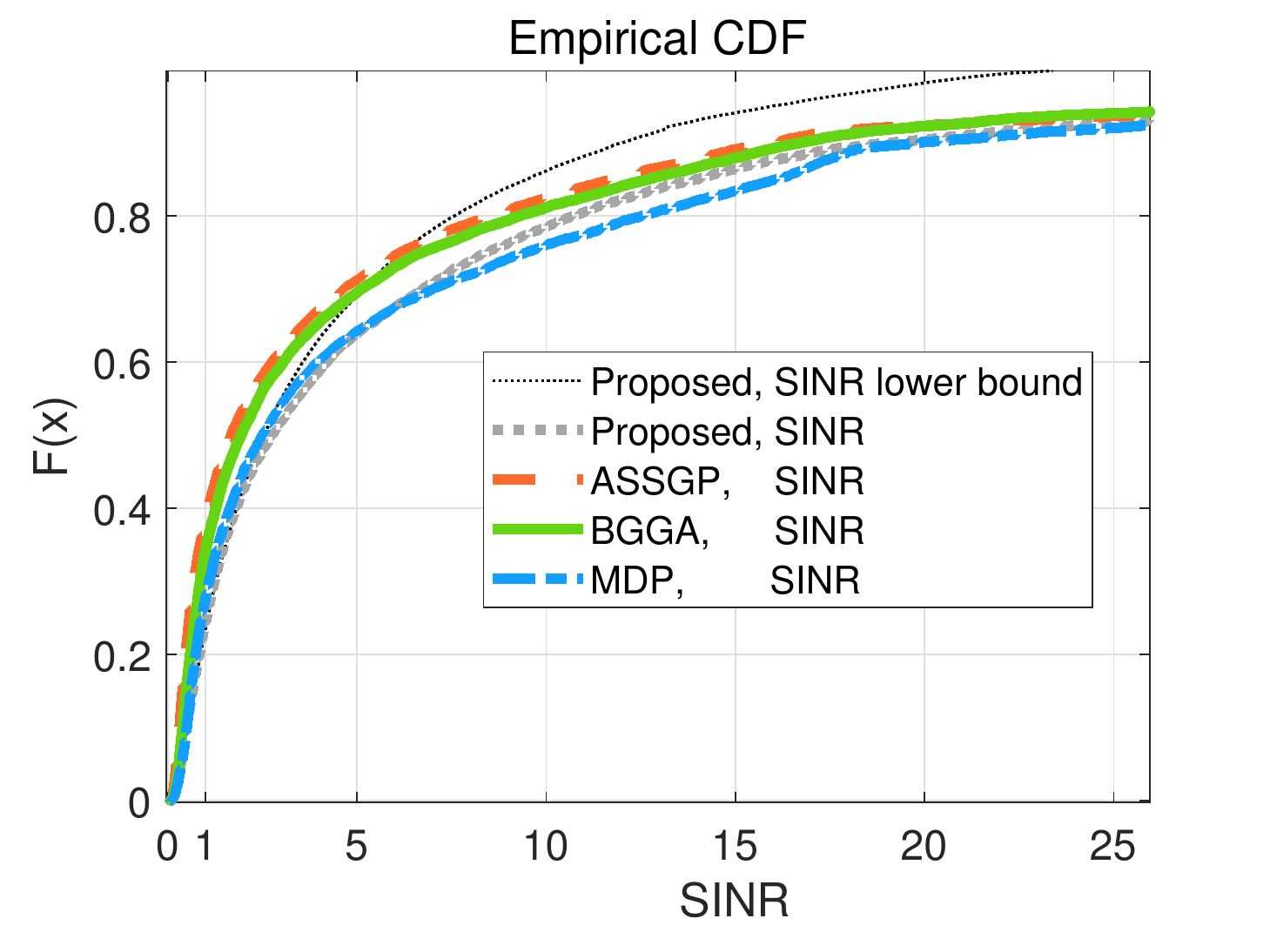}
 %\vspace{-1.5cm}
  \caption{CDF of the links' SINR and its SINR lower bound $\overline{\text{SINR}}_{b,g}(\mathbf{y}^{\star},\mathbf{X}^{\star})$ in \eqref{SINR_Exp}.}
  \label{fig:SINRCDF}
\end{figure}

\subsubsection{SINR Performance}
 \tcg{For each link from BS $b$ to grid $g$ ($x_{b,g}=1$), we collect the $\text{SINR}_{b,g}(\mathbf{y}^{\star},\mathbf{X}^{\star})$ values for $50$ random realizations for the results in Fig.\ref{BSDeployResults}. 
 The cumulative distribution functions (CDFs) of the collected $\left\{\text{SINR}_{b,g}(\mathbf{y}^{\star},\mathbf{X}^{\star})\right\}$ of the proposed and benchmark schemes are displayed  in Fig.~\ref{fig:SINRCDF}, where $F(x)$  denotes the CDF of the data that is smaller than or equal to the abscissa $x$. For the proposed scheme, the CDF of the SINR lower bound $\overline{\text{SINR}}_{b,g}(\mathbf{y}^{\star},\mathbf{X}^{\star})$ in \eqref{SINR_Exp} is also plotted. 
It is observed that the MDP reveals the best SINR performance due to the smallest number of deployed BSs and thus lower network interference level.
However, this outpacing result is due to the ignorance of the physical and UE access-limited blockages, resulting in the worst UE outage performance as will be shown in Fig. \ref{fig:outage_statistic}. 
The proposed scheme, deploying $89$ BSs, has a slightly larger number of the deployed BSs than the ASSGP ($84$ BSs),  but has better SINR performance since ASSGP does not account for the SINR outdage during its deployment. The BGGA that deploys the largest number of BSs (i.e., $102$ BSs) reveals a worse SINR performance than that of the proposed scheme due to the increased network interference. The SINR outage for a connected link occurs when the SINR is lower than a threshold $z=1$ in TABLE \ref{tab:1}. It is noticed from Fig.~\ref{fig:SINRCDF} that $77\%$ of the links have $\overline{\text{SINR}}_{b,g}(\mathbf{y}^{\star},\mathbf{X}^{\star})\geq z=1$, while $78\%$ of the links have ${\text{SINR}}_{b,g}(\mathbf{y}^{\star},\mathbf{X}^{\star})\geq z=1$. 
This reveals that the SINR outage upper bound in \eqref{SINRoutage_upperbound} (i.e., the SINR lower bound in \eqref{SINR_Exp}) is tight for at least $77\%$ links. However, seen from Fig.~\ref{fig:SINRCDF}, the gap between the bound and true value grows as the SINR  increases. }

%Since the $\text{SINR}_{b,g}(\mathbf{y}^{\star},\mathbf{X}^{\star})$ is a random variable due to the random interference,
%$\Pr\left( \overline{\text{SINR}}_{b,g}(\mathbf{y},\mathbf{X}) \leq z \right)\approx\Pr\left( {\text{SINR}}_{b,g}(\mathbf{y},\mathbf{X}) \leq z \right) =23\% $

%\taejoon{I must admit that I cannot understand the last two sentences below. Can you mark the $77$\%  line in the figure? Why we should look at $77$\%? Why not for example $80$\%? Why $z=1$?  How the right-hand-side of the equation below is zero? Sorry but these are the raised question from this two sentences. You would also receive the similar questions from the reviwers.}It is also noticed that $77\%$ SINR lower bound $\overline{\text{SINR}}_{b,g}(\mathbf{y},\mathbf{X})$ of the proposed scheme is larger  than the threshold %$z=1$, which leads to 
%$$\Pr\!\left({\text{SINR}}_{b,g}(\mathbf{y},\mathbf{X})\!<\! z| \mathbf{y},\mathbf{X}\!\right) \!\leq\!  \Pr\!\left(\overline{\text{SINR}}_{b,g}(\mathbf{y},\mathbf{X})\!<\! z| \mathbf{y},\mathbf{X}\!\right)\!=\!0.$$
%This indicates that the SINR outage probability lower bound in \eqref{eqn:pUB} is tight with probability no less than $77\%$ \taejoon{I cannot clearly understand how could you derive this conclusion from the figure.}. 

\begin{table*}[h]
\caption{Proposed BS deployment under different parameter settings.
}
\centering
\begin{tabular}{|c|l|c|c|c|}
\hline
 Different $B$ & Parameter setting                                                     & Number of BSs   & Deployment cost & Running time (minutes)  \\ 
\hline
$B=240$  & $N_{RF}\!=\!12$, $\gamma\!=\!0.05, \zeta\!=\!0.05$    & 89  & 48.4 & 45  \\ 
\hline
$B=240$ & $N_{RF}\!=\!14$,
$\gamma\!=\!0.05, \zeta\!=\!0.05$    & 88  & 47.6 & 77   \\
\hline
$B=240$ & $N_{RF}\!=\!14$, $\gamma\!=\!0.05, \zeta\!=\!0.1$               & 82         & 44.6 & 78     \\
\hline
$B=184$ & $N_{RF}\!=\!12$, $\gamma\!=\!0.05, \zeta\!=\!0.05$          & Infeasible    & Not available & Not available  \\
\hline
$B=184$ & $N_{RF}\!=\!14$, $\gamma\!=\!0.05, \zeta\!=\!0.05$               & Infeasible       & Not available & Not available        \\
\hline
$B=184$ & $N_{RF}\!=\!14$, $\gamma\!=\!0.05, \zeta\!=\!0.1$               & 84       & 45.4 & 51        \\
\hline
$B=130$ & $N_{RF}\!=\!12$, $\gamma\!=\!0.05, \zeta\!=\!0.05$          & Infeasible   & Not available & Not available  \\
\hline
$B=130$ & $N_{RF}\!=\!14$, $\gamma\!=\!0.05, \zeta\!=\!0.1$               & Infeasible      & Not available & Not available         \\
\hline
$B=130$ & $N_{RF}\!=\!14$, $\gamma\!=\!0.05, \zeta\!=\!0.2$               & 50    & 21.2 & 19           \\
\hline
\end{tabular}
\label{tab:ProposedResult_diffSet}
\end{table*}

\subsubsection{Varying Number of Candidate BS Locations}
In TABLE~\ref{tab:ProposedResult_diffSet}, we present the results of the proposed BS deployment for different numbers of candidate BS locations as in Fig.~\ref{fig:geometry2} and for different sets of  parameters.
It can be observed that increasing the number of RF chains $N_{\text{RF}}$ decreases  the number of deployed BSs. This is because a BS with a larger $N_{\text{RF}}$ can afford  a larger $\Phi$ value in  \eqref{Re_NumUE_PerBS} and thereby, covers more grids. \tcg{Moreover, it is noticed that the number of candidate BS locations $B$ also impact to the BS deployment results.
When $N_{RF}\!=\!14$, $\gamma\!=\!0.05, \zeta\!=\!0.1$ in TABLE~\ref{tab:ProposedResult_diffSet}, two more BSs are deployed when $B=184$ due to the reduced search space for BS deployment compared to the case  when $B=240$. 
However, as we reduce the number of candidate BS locations $B$, it raises the infeasibility issue of the proposed BS deployment scheme as shown in TABLE~\ref{tab:ProposedResult_diffSet}. }

\subsubsection{Time Complexity} TABLE~\ref{tab:ProposedResult_diffSet} also displays the time complexity of the  proposed scheme. The time overhead is measured in minutes using Gurobi \cite{Gurobi}.
Compared to the time complexities of MDP, ASSGP, BGGA when $B=240, N_{RF}\!=\!12$, $\gamma\!=\!0.05, \zeta\!=\!0.05$, whose running times are $1$ minutes, $5$ minutes, and $12$ minutes, respectively, time complexity of the proposed scheme is exceedingly high. However, considering the fact that  our proposed scheme optimally solves the INP in \eqref{Opt_problem} and runs off-line, it is not a serious drawback. As aforementioned, solving the INP in \eqref{Opt_problem} for the large-scale problem size in TABLE~\ref{tab:1}  ($B\times G\!=\!1,774,320$) is difficult if not impossible. \tcg{Directly solving them using available solvers often encounters memory outage or never-ending running time.
Although the proposed linearization technique in Lemma~\ref{Const_linearization} introduces twice more additional variables $s_{b,g}, \forall b, g$ than the benchmarks, implementing the variable nulling strategy in Section~\ref{newVar_lowComp} effectively alleviate the computational overhead.}

\begin{figure}
  \centering
  \includegraphics[width=0.5\textwidth]{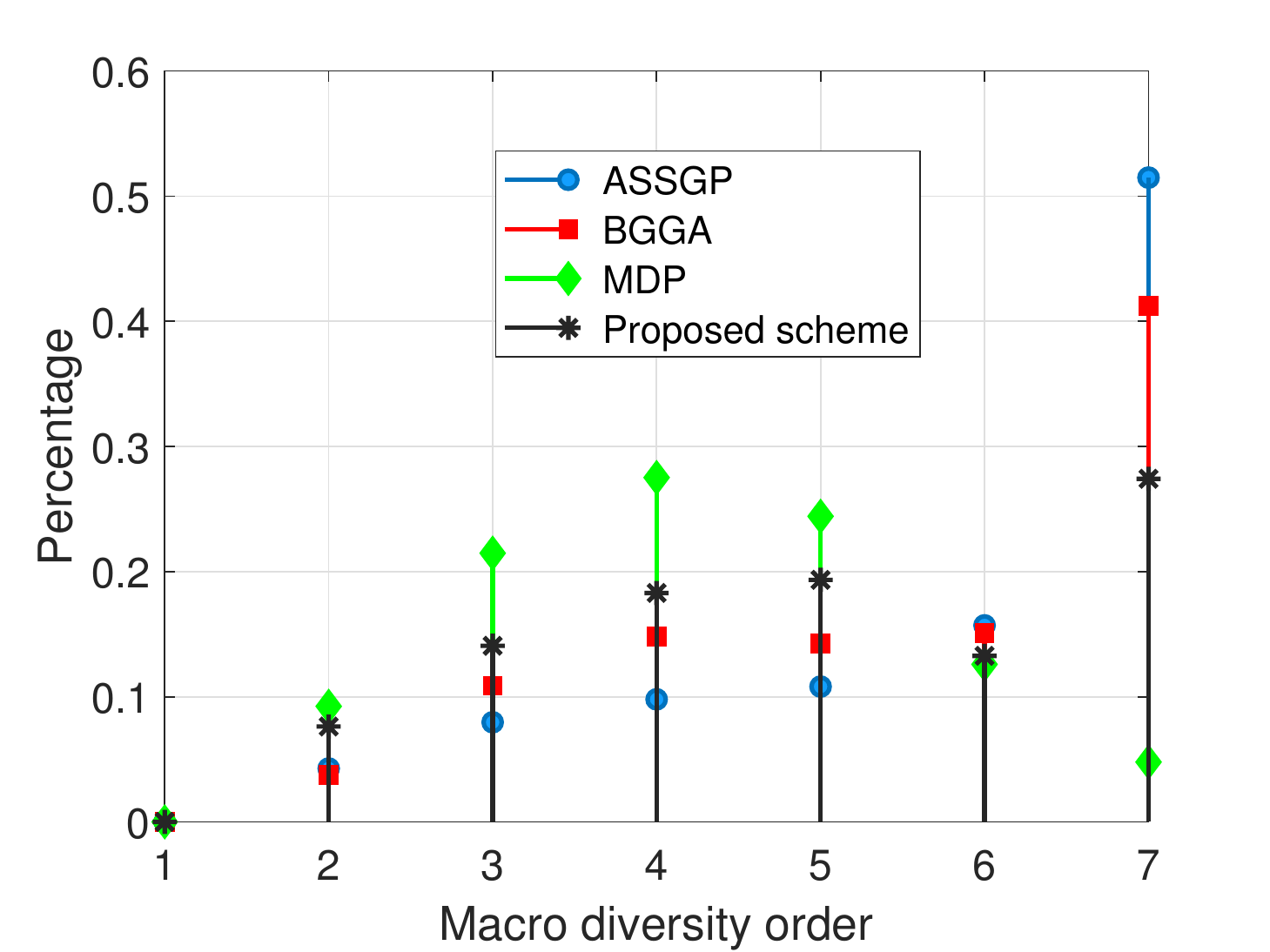}
  %\vspace{-1.5cm}
  \caption{Probabilities of the macro diversity order for the MDP, ASSGP, BGGA, and proposed scheme.}
  \label{fig:NumAssoci_UE_statistic}
\end{figure}

\subsubsection{Macro Diversity Order Distribution} The macro diversity orders  $\sum_{b\in\mathcal{B}}  x_{b,g},\! \forall g\!\in\!\mathcal{G}$ in \eqref{macro_diversity}
of each scheme  are collected and  $\Pr\left(\sum_{b} x_{b,g}=i  \right)$ are presented in Fig.~\ref{fig:NumAssoci_UE_statistic} for the deployment results in Fig.~\ref{BSDeployResults}.
Note that all schemes guarantee a minimum macro diversity order $2$, which is a constraint for the MDP and ASSGP, and an implicit requirement for the BGGA and proposed scheme for UE outage mitigation.  
Without the UE access-limited blockage constraint, deployed BSs of MDP can cover any LoS-visible grids within $R^{\max}$ and hence it deploys the minimum number of BSs (i.e., $61$ BSs) to produce the largest $\Pr\left(\sum_{b} x_{b,g}=i  \right)$  at $i=2,3,4,5$ as seen in Fig.~\ref{fig:NumAssoci_UE_statistic}.   
%A BS which is deployed to cover grid $a$ will also cover grid $b$ even though the coverage constraint  of grid $b$ has already been satisfied. This is why the macro diversity order $\sum_{b} x_{b,g}=6$ and $\sum_{b} x_{b,g}=7$ appear in Fig.~\ref{fig:NumAssoci_UE_statistic}.
While the proposed scheme has a similar (respectively, smaller) number of deployed BSs to the ASSGP (than the BGGA), its $\Pr\left(\sum_{b} x_{b,g}=i  \right)$  at $i=2,3,4,5$ is larger than those of the ASSGP and BGGA, which demonstrates the superior performance of the proposed scheme compared to the benchmarks in terms of providing UE outage guarantees; this will be clear in Fig.~\ref{fig:outage_statistic}.

\subsubsection{UE Access-Limited Blockage Probability}
In Fig. \ref{fig:CapBlock_statistic}, given the
BS deployment results in Fig.~\ref{BSDeployResults}, we collect the UE access-limited blockage probabilities for each BS and demonstrate the CDFs  of 
$ \left\{\rho_{b, g}(\bx_b) \; \text{in}\; \eqref{express2} \;\; \forall b\in\cB \; \text{with}\; y_{b}\!=\!1 \right\}$.
It can be observed from the curves that the proposed scheme and BGGA deploy the BSs to guarantee that each UE's access-limited blockage probability is limited by the tolerance $\gamma\!=\!0.05$ in TABLE \ref{tab:1}.
However, around $30\sim 40\%$ of the deployed BSs with the MDP and ASSGP schemes have the UE access-limited blockage probability larger than $0.05$. \tcg{It should be emphasized that this stark guarantee is achieved by deploying $89$ BSs of the proposed scheme, while the MDP, ASSGP, and BGGA deploy $61$, $84$, and $102$ BSs, respectively.}

\begin{figure}[htbp]
  \centering
  \includegraphics[width=0.5\textwidth]{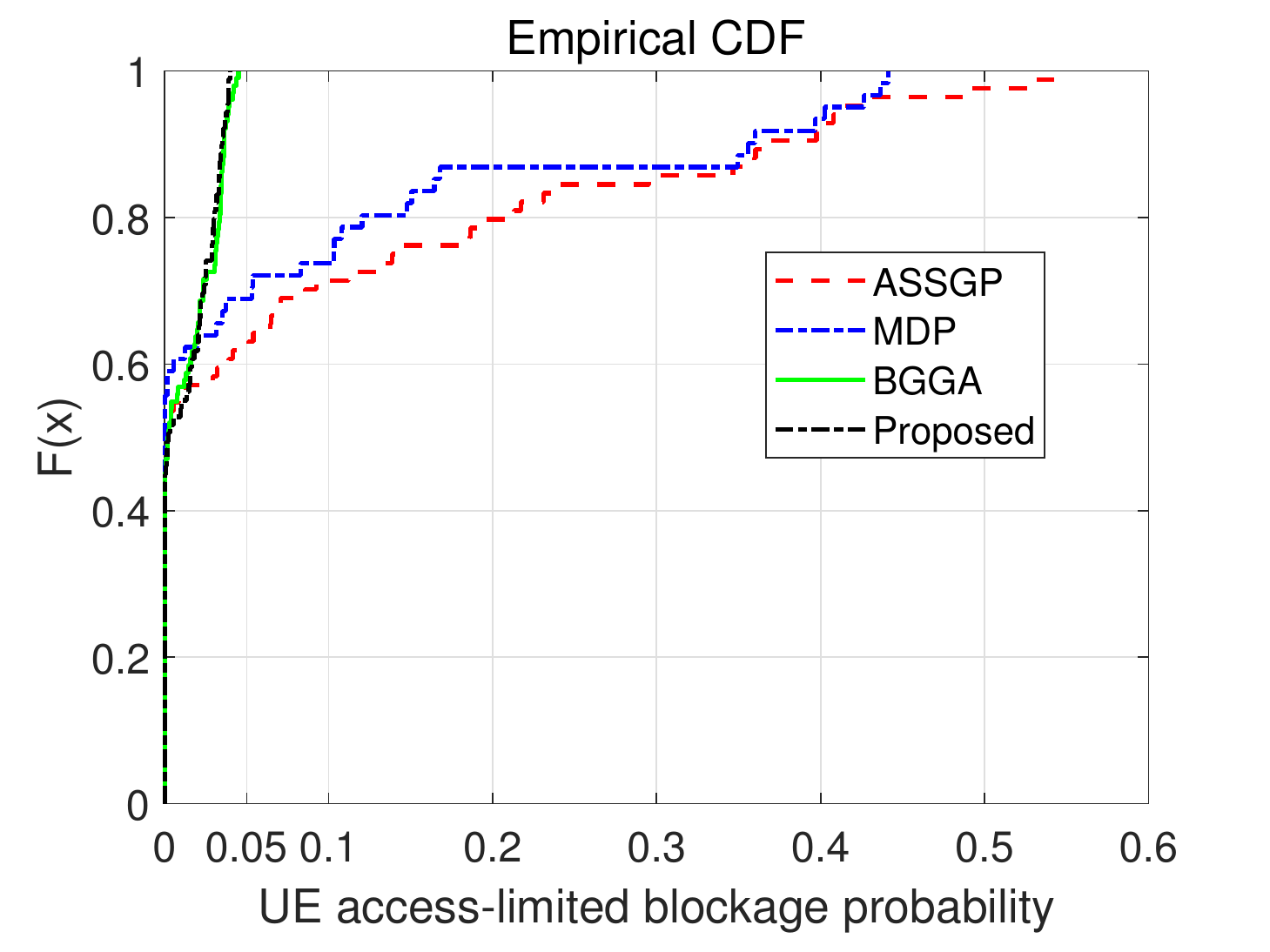}
  \caption{CDF of the collected UE access-limited blockage statistics, i.e., CDFs of  $\left\{\rho_{b, g}(\bx_b) \; \text{in}\; \eqref{express2} \;\; \forall b\in\cB \; \text{with}\; y_{b}\!=\!1 \right\}$.}
  \label{fig:CapBlock_statistic}
\end{figure}

\subsubsection{UE Outage Probability} \label{subsubsec:UE Outage}

\tcg{Based on the BS deployment results in Fig.~\ref{BSDeployResults}, we collect the UE outage probabilities in \eqref{UE_outage}.
The CDFs of the collected UE outage probabilities are demonstrated in Fig.~\ref{fig:outage_statistic}.
It is evident that the proposed scheme  guarantees the UE outage probability with the specified tolerance $\zeta\!=\!0.05$ in TABLE \ref{tab:1}. Even through the ASSGP  deploys the similar number of BSs as the proposed scheme, its UE outage performance is much worse than the proposed scheme and nearly $12\%$ UEs have outage probability larger than $0.1$.
The BGGA exhibits the similar UE outage statistics to the proposed scheme. However, it fails to provides the guarantee and its performance is achieved by deploying $13$-more BSs than the proposed scheme. }

\begin{figure}
  \centering
  \includegraphics[width=0.5\textwidth]{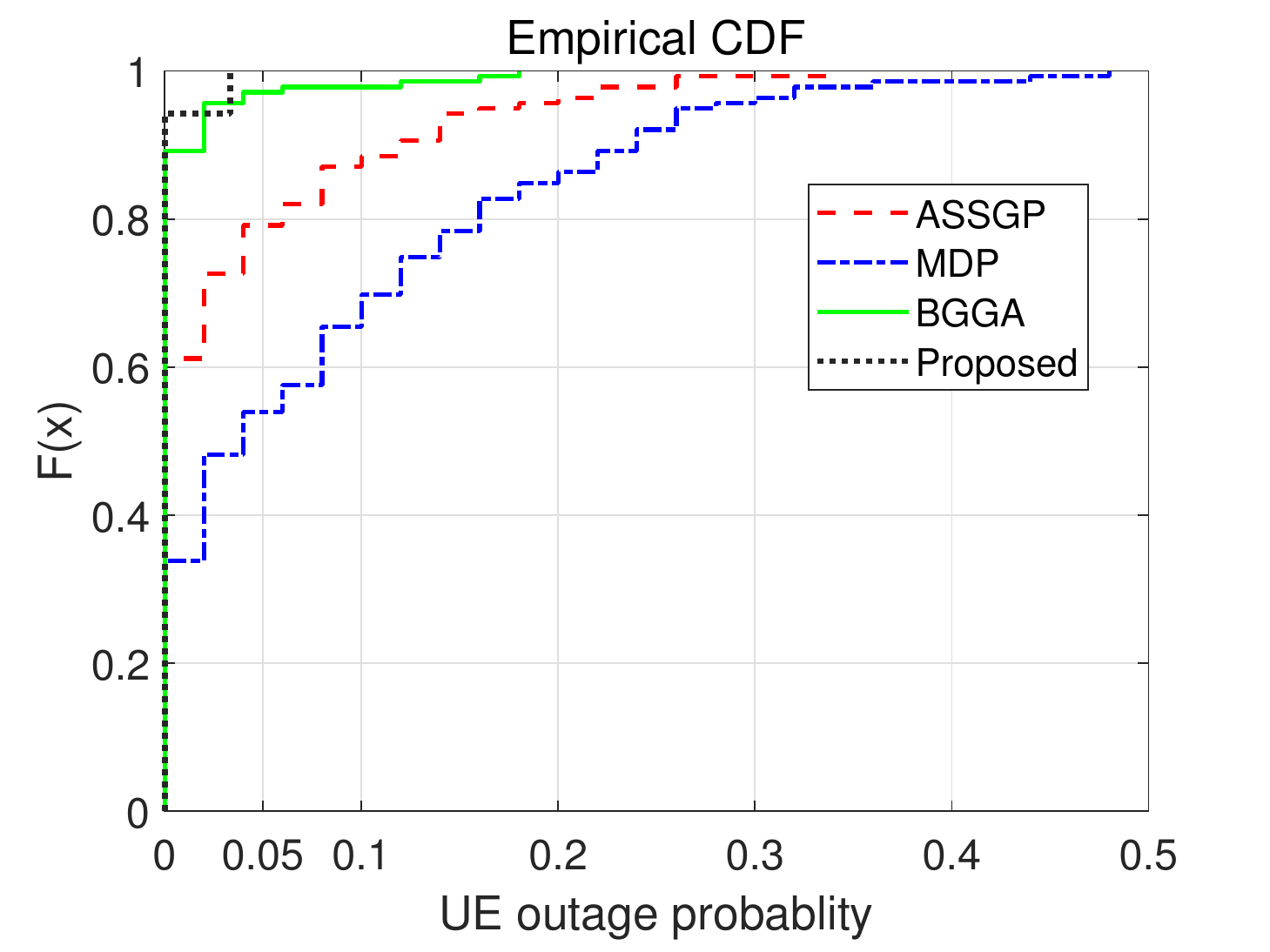}
  \caption{CDF of the collected UE outage probabilities.}
  \label{fig:outage_statistic}
\end{figure}

%%================================================================================================
%%================================================================================================
%%================================================================================================
%%================================================================================================
% \section{BS Deployment Approaches}~\label{sec:two_methods}
\section{Conclusions}    \label{Sec_con}
%%================================================================================================
%%================================================================================================
%%================================================================================================
%%================================================================================================
We addressed important mmWave connectivity challenges in a 3D urban geometry by proposing  a link quality-guaranteed minimum-cost mmWave BS deployment scheme that  jointly optimizes the {BS placement} and {cell coverage}.
To mathematically formulate the problem, we first introduced the stochastic mmWave link state model and used it to characterize the BS association and UE outage constraints.
The BS deployment problem was then formulated as INP, which was optimally solved by decomposing it into two separable subproblems: (i) BS coverage optimization problem and (ii) minimum subset BS selection problem.
\tcg{We provided the optimal solutions for these subproblems as well as their theoretical justifications.
Simulation results demonstrated the efficacy of the proposed scheme in terms of the BS deployment cost,  computational complexity, UE access-limited blockage, and UE outage performance. Compared to the MDP, ASSGP, and BGGA benchmarks, our proposed algorithm provides guaranteed tolerance to UE access-limited blockage and UE outage. It should be noted here that our main goal in this work was to study the principle of minimum-cost BS deployment for combined coverage and link quality constraints in mmWave networks, and through simulations describe the gain that can be expected by taking on such an approach.
One major drawback of the proposed scheme was that the time complexity is exceedingly high compared to other benchmarks.
However, considering the fact that the BS deployment planning is done off-line in practice, our proposed scheme optimally solves the INP in \eqref{Opt_problem}, and the proposed scheme provided stark outage guarantees, the high time complexity is not a serious drawback.}

\small
\bibliographystyle{IEEEtran}
\bibliography{IEEEabrv,survey_ref_phy}

% Generated by IEEEtran.bst, version: 1.13 (2008/09/30)
\begin{thebibliography}{10}
\providecommand{\url}[1]{#1}
\csname url@samestyle\endcsname
\providecommand{\newblock}{\relax}
\providecommand{\bibinfo}[2]{#2}
\providecommand{\BIBentrySTDinterwordspacing}{\spaceskip=0pt\relax}
\providecommand{\BIBentryALTinterwordstretchfactor}{4}
\providecommand{\BIBentryALTinterwordspacing}{\spaceskip=\fontdimen2\font plus
\BIBentryALTinterwordstretchfactor\fontdimen3\font minus
  \fontdimen4\font\relax}
\providecommand{\BIBforeignlanguage}[2]{{%
\expandafter\ifx\csname l@#1\endcsname\relax
\typeout{** WARNING: IEEEtran.bst: No hyphenation pattern has been}%
\typeout{** loaded for the language `#1'. Using the pattern for}%
\typeout{** the default language instead.}%
\else
\language=\csname l@#1\endcsname
\fi
#2}}
\providecommand{\BIBdecl}{\relax}
\BIBdecl

\bibitem{Kims_gen}
S.~Hur, T.~Kim, D.~J. Love, J.~V. Krogmeier, T.~A. Thomas, and A.~Ghosh,
  ``Millimeter wave beamforming for wireless backhaul and access in small cell
  networks,'' \emph{IEEE Transactions on Communications}, vol.~61, no.~10, pp.
  4391--4403, Oct. 2013.

\bibitem{mmWave_magzine}
B.~P.~S. {Sahoo}, C.~{Chou}, C.~{Weng}, and H.~{Wei}, ``Enabling
  millimeter-wave 5g networks for massive {IoT} applications: A closer look at
  the issues impacting millimeter-waves in consumer devices under the {5G}
  framework,'' \emph{IEEE Consumer Electronics Magazine}, vol.~8, no.~1, pp.
  49--54, Jan 2019.

\bibitem{mmWave_magzine2}
V.~{Raghavan}, A.~{Partyka}, A.~{Sampath}, S.~{Subramanian}, O.~H. {Koymen},
  K.~{Ravid}, J.~{Cezanne}, K.~{Mukkavilli}, and J.~{Li}, ``Millimeter-wave
  {MIMO} prototype: Measurements and experimental results,'' \emph{IEEE
  Communications Magazine}, vol.~56, no.~1, pp. 202--209, Jan 2018.

\bibitem{HealmmWave}
R.~W. Heath, N.~González-Prelcic, S.~Rangan, W.~Roh, and A.~M. Sayeed, ``An
  overview of signal processing techniques for millimeter wave {MIMO}
  systems,'' \emph{IEEE Journal of Selected Topics in Signal Processing},
  vol.~10, no.~3, pp. 436--453, 2016.

\bibitem{Hadi2016}
H.~{Ghauch}, T.~{Kim}, M.~{Bengtsson}, and M.~{Skoglund}, ``Subspace estimation
  and decomposition for large millimeter-wave mimo systems,'' \emph{IEEE
  Journal of Selected Topics in Signal Processing}, vol.~10, no.~3, pp.
  528--542, April 2016.

\bibitem{BSCap_limit2}
A.~{Alkhateeb}, R.~W. {Heath}, and G.~{Leus}, ``Achievable rates of multi-user
  millimeter wave systems with hybrid precoding,'' in \emph{2015 IEEE
  International Conference on Communication Workshop (ICCW)}, Jun. 2015, pp.
  1232--1237.

\bibitem{Zhang18}
W.~{Zhang}, T.~{Kim}, D.~J. {Love}, and E.~{Perrins}, ``Leveraging the
  restricted isometry property: Improved low-rank subspace decomposition for
  hybrid millimeter-wave systems,'' \emph{IEEE Transactions on Communications},
  vol.~66, no.~11, pp. 5814--5827, Nov 2018.

\bibitem{aych_RFChain}
O.~E. Ayach, S.~Rajagopal, S.~Abu-Surra, Z.~Pi, and R.~W. Heath, ``Spatially
  sparse precoding in millimeter wave {MIMO} systems,'' \emph{IEEE Transactions
  on Wireless Communications}, vol.~13, no.~3, pp. 1499--1513, Mar. 2014.

\bibitem{Alkhateeb}
A.~Alkhateeb, O.~E. Ayach, G.~Leus, and R.~W. Heath, ``Channel estimation and
  hybrid precoding for millimeter wave cellular systems,'' \emph{IEEE Journal
  of Selected Topics in Signal Processing}, vol.~8, no.~5, pp. 831--846, Oct.
  2014.

\bibitem{Link_status}
T.~S. {Rappaport}, Y.~{Xing}, G.~R. {MacCartney}, A.~F. {Molisch},
  E.~{Mellios}, and J.~{Zhang}, ``Overview of millimeter wave communications
  for fifth-generation ({5G}) wireless networks—with a focus on propagation
  models,'' \emph{IEEE Transactions on Antennas and Propagation}, vol.~65,
  no.~12, pp. 6213--6230, 2017.

\bibitem{RobertHealth1}
T.~Bai, R.~Vaze, and R.~W. Heath, ``Analysis of blockage effects on urban
  cellular networks,'' \emph{IEEE Transactions on Wireless Communications},
  vol.~13, no.~9, pp. 5070--5083, Sep. 2014.

\bibitem{Mac_div2}
M.~Dong and T.~Kim, ``Interference analysis for millimeter-wave networks with
  geometry-dependent first-order reflections,'' \emph{IEEE Transactions on
  Vehicular Technology}, vol.~67, no.~12, pp. 12\,404--12\,409, Dec. 2018.

\bibitem{Mac_div1}
J.~Choi, ``On the macro diversity with multiple {BSs} to mitigate blockage in
  millimeter-wave communications,'' \emph{IEEE Communications Letters},
  vol.~18, no.~9, pp. 1653--1656, Sep. 2014.

\bibitem{UEAssoc_2}
A.~{Alizadeh} and M.~{Vu}, ``Time-fractional user association in millimeter
  wave {MIMO} networks,'' in \emph{2018 IEEE International Conference on
  Communications (ICC)}, May 2018, pp. 1--6.

\bibitem{UEAssoc_3}
H.~{Zhang}, S.~{Huang}, C.~{Jiang}, K.~{Long}, V.~C.~M. {Leung}, and H.~V.
  {Poor}, ``Energy efficient user association and power allocation in
  millimeter-wave-based ultra dense networks with energy harvesting base
  stations,'' \emph{IEEE Journal on Selected Areas in Communications}, vol.~35,
  no.~9, pp. 1936--1947, Sep. 2017.

\bibitem{whitepapaer_smallcell}
\BIBentryALTinterwordspacing
(2019, Oct) Precision planning for {5G} era network with smallcells, white
  paper. [Online]. Available:
  \url{https://www.scf.io/en/documents/230_Precision_planning_for_5G_Era_networks_with_small_cells.php}
\BIBentrySTDinterwordspacing

\bibitem{Stoc_BS1}
J.~Peng, P.~Hong, and K.~Xue, ``Energy-aware cellular deployment strategy under
  coverage performance constraints,'' \emph{IEEE Transactions on Wireless
  Communications}, vol.~14, no.~1, pp. 69--80, Jan. 2015.

\bibitem{Stoc_BS1_ZhaoXidian}
G.~Zhao, S.~Chen, L.~Zhao, and L.~Hanzo, ``Joint energy-spectral-efficiency
  optimization of {CoMP} and {BS} deployment in dense large-scale cellular
  networks,'' \emph{IEEE Transactions on Wireless Communications}, vol.~16,
  no.~7, pp. 4832--4847, Jul. 2017.

\bibitem{Stoc_BS3}
B.~Yang, G.~Mao, X.~Ge, M.~Ding, and X.~Yang, ``On the energy-efficient
  deployment for ultra-dense heterogeneous networks with {NLoS} and {LoS}
  transmissions,'' \emph{IEEE Transactions on Green Communications and
  Networking}, vol.~2, no.~2, pp. 369--384, Jun. 2018.

\bibitem{Stoc_BS4}
P.~Mekikis, E.~Kartsakli, A.~Antonopoulos, A.~S. Lalos, L.~Alonso, and
  C.~Verikoukis, ``Two-tier cellular random network planning for minimum
  deployment cost,'' in \emph{IEEE International Conference on Communications},
  Jun. 2014, pp. 1248--1253.

\bibitem{Stoc_BS5}
C.~Peng, L.~Wang, and C.~Liu, ``Optimal base station deployment for small cell
  networks with energy-efficient power control,'' in \emph{2015 IEEE
  International Conference on Communications (ICC)}, Jun. 2015, pp. 1863--1868.

\bibitem{my_Access}
M.~{Dong}, T.~{Kim}, J.~{Wu}, and E.~W.~M. {Wong}, ``Cost-efficient millimeter
  wave base station deployment in manhattan-type geometry,'' \emph{IEEE
  Access}, vol.~7, pp. 149\,959--149\,970, 2019.

\bibitem{MinBSnum_VT}
S.~{Chatterjee}, M.~J. {Abdel-Rahman}, and A.~B. {MacKenzie}, ``Optimal base
  station deployment with downlink rate coverage probability constraint,''
  \emph{IEEE Wireless Communications Letters}, vol.~7, no.~3, pp. 340--343,
  Jun. 2018.

\bibitem{sub6_siteBS3}
C.~Fan, T.~Zhang, and Z.~Zeng, ``Energy-efficient base station deployment in
  {H}et{N}et based on traffic load distribution,'' in \emph{2017 IEEE 85th
  Vehicular Technology Conference (VTC Spring)}, Jun. 2017, pp. 1--5.

\bibitem{sub6_siteBS2}
M.~A. Yigitel, O.~D. Incel, and C.~Ersoy, ``Dynamic {BS} topology management
  for green next generation {H}et{N}ets: An urban case study,'' \emph{IEEE
  Journal on Selected Areas in Communications}, vol.~34, no.~12, pp.
  3482--3498, Dec. 2016.

\bibitem{sub6_siteBS1}
C.~C. Coskun and E.~Ayanoglu, ``Energy-efficient base station deployment in
  heterogeneous networks,'' \emph{IEEE Wireless Communications Letters},
  vol.~3, no.~6, pp. 593--596, Dec. 2014.

\bibitem{Site_BS2}
Y.~Lu, H.-W. Hsu, and L.-C. Wang, ``Performance model and deployment strategy
  for mm-wave multi-cellular systems,'' in \emph{2016 25th Wireless and Optical
  Communication Conference}, May 2016, pp. 1--4.

\bibitem{Site_BS3}
S.~S. Szyszkowicz, A.~Lou, and H.~Yanikomeroglu, ``Automated placement of
  individual millimeter-wave wall-mounted base stations for line-of-sight
  coverage of outdoor urban areas,'' \emph{IEEE Wireless Communications
  Letters}, vol.~5, no.~3, pp. 316--319, Jun. 2016.

\bibitem{Site_BS4}
N.~Palizban, S.~Szyszkowicz, and H.~Yanikomeroglu, ``Automation of millimeter
  wave network planning for outdoor coverage in dense urban areas using
  wall-mounted base stations,'' \emph{IEEE Wireless Communications Letters},
  vol.~6, no.~2, pp. 206--209, Apr. 2017.

\bibitem{7794888}
M.~{Gonzalez} and J.~{Thompson}, ``An energy efficient base station deployment
  for mm-wave based wireless backhaul,'' in \emph{2016 IEEE 27th Annual
  International Symposium on Personal, Indoor, and Mobile Radio Communications
  (PIMRC)}, Sep. 2016, pp. 1--6.

\bibitem{ZhangYue}
Y.~Zhang, L.~Dai, and E.~W.~M. Wong, ``Optimal {BS} deployment and user
  association for {5G} millimeter wave communication networks,'' \emph{IEEE
  Transactions on Wireless Communications}, vol.~20, no.~5, pp. 2776--2791,
  2021.

\bibitem{mmWaveSiteDeploy}
\BIBentryALTinterwordspacing
I.~Mavromatis, A.~Tassi, R.~J. Piechocki, and A.~R. Nix, ``Efficient
  millimeter-wave infrastructure placement for city-scale {ITS},'' \emph{CoRR},
  vol. abs/1903.01372, 2019. [Online]. Available:
  \url{http://arxiv.org/abs/1903.01372}
\BIBentrySTDinterwordspacing

\bibitem{my_PartialSamplingBS}
M.~{Dong}, T.~{Kim}, J.~{Wu}, and W.~M.~E. {Wong}, ``Millimeter-wave base
  station deployment using the scenario sampling approach,'' \emph{IEEE
  Transactions on Vehicular Technology}, vol.~69, no.~11, pp. 14\,013--14\,018,
  2020.

\bibitem{MmWaveBsDeploy_UEOre}
M.~Naderi~Soorki, W.~Saad, and M.~Bennis, ``Optimized deployment of millimeter
  wave networks for in-venue regions with stochastic users’ orientation,''
  \emph{IEEE Transactions on Wireless Communications}, vol.~18, no.~11, pp.
  5037--5049, 2019.

\bibitem{sub6_siteBS5}
K.~Shen, Y.~Liu, D.~Y. Ding, and W.~Yu, ``Flexible multiple base station
  association and activation for downlink heterogeneous networks,'' \emph{IEEE
  Signal Processing Letters}, vol.~24, no.~10, pp. 1498--1502, Oct. 2017.

\bibitem{lagrangian_dual}
M.~{Feng}, S.~{Mao}, and T.~{Jiang}, ``{BOOST}: Base station on-off switching
  strategy for green massive {MIMO} hetnets,'' \emph{IEEE Transactions on
  Wireless Communications}, vol.~16, no.~11, pp. 7319--7332, Nov. 2017.

\bibitem{sub6_siteBS4}
X.~Lin and S.~Wang, ``Joint user association and base station switching on/off
  for green heterogeneous cellular networks,'' in \emph{2017 IEEE International
  Conference on Communications (ICC)}, May 2017, pp. 1--6.

\bibitem{BSSleep_Femto}
J.~Kim, W.~S. Jeon, and D.~G. Jeong, ``Base-station sleep management in
  open-access femtocell networks,'' \emph{IEEE Transactions on Vehicular
  Technology}, vol.~65, no.~5, pp. 3786--3791, May 2016.

\bibitem{NonPPP1}
J.~Vales-Alonso, F.~Parrado-García, P.~López-Matencio, J.~Alcaraz, and
  F.~González-Castaño, ``On the optimal random deployment of wireless sensor
  networks in non-homogeneous scenarios,'' \emph{Ad Hoc Networks}, vol.~11,
  no.~3, pp. 846--860, 2013.

\bibitem{NonPPP2}
M.~Haenggi, ``Stochastic geometry for wireless networks,'' \emph{Cambridge,
  U.K.: Cambridge Univ. Press}, 2012.

\bibitem{MyRay_tracing}
M.~Dong, W.~Chan, T.~Kim, K.~Liu, H.~Huang, and G.~Wang, ``Simulation study on
  millimeter wave 3{D} beamforming systems in urban outdoor multi-cell
  scenarios using 3{D} ray tracing,'' in \emph{IEEE 26th Annual International
  Symposium on PIMRC}, Aug. 2015, pp. 2265--2270.

\bibitem{3GPP38211}
\BIBentryALTinterwordspacing
G.~T. 38.211. (2020, Jan.) Physical channels and modulation. [Online].
  Available:
  \url{https://portal.3gpp.org/desktopmodules/Specifications/SpecificationDetails.aspx?specificationId=3213}
\BIBentrySTDinterwordspacing

\bibitem{Link_distance}
T.~S. Rappaport, S.~Sun, R.~Mayzus, H.~Zhao, Y.~Azar, K.~Wang, G.~N. Wong,
  J.~K. Schulz, M.~Samimi, and F.~Gutierrez, ``Millimeter wave mobile
  communications for 5{G} cellular: It will work!'' \emph{IEEE Access}, vol.~1,
  pp. 335--349, May 2013.

\bibitem{RobertHealth2}
T.~Bai and R.~W. Heath, ``Coverage and rate analysis for millimeter-wave
  cellular networks,'' \emph{IEEE Transactions on Wireless Communications},
  vol.~14, no.~2, pp. 1100--1114, Feb. 2015.

\bibitem{BeamPattern}
M.~D. Renzo, ``Stochastic geometry modeling and analysis of multi-tier
  millimeter wave cellular networks,'' \emph{IEEE Transactions on Wireless
  Communications}, vol.~14, no.~9, pp. 5038--5057, Sep. 2015.

\bibitem{twoBall_model}
------, ``Stochastic geometry modeling and analysis of multi-tier millimeter
  wave cellular networks,'' \emph{IEEE Transactions on Wireless
  Communications}, vol.~14, no.~9, pp. 5038--5057, Sep. 2015.

\bibitem{thinning}
R.~L. Strei, \emph{Poisson Point Processess: Imaging, Tracking, and
  Sensing}.\hskip 1em plus 0.5em minus 0.4em\relax Springer US, Sep. 2010.

\bibitem{INP1}
D.~Li and X.~Sun, \emph{Nonlinear Integer Programming}.\hskip 1em plus 0.5em
  minus 0.4em\relax Springer US, 2006.

\bibitem{Gurobi}
\BIBentryALTinterwordspacing
Gurobi. (2018) Gurobi optimier quick start guide. [Online]. Available:
  \url{https://www.gurobi.com/wp-content/plugins/hd_documentations/content/pdf/quickstart_windows_8.1.pdf}
\BIBentrySTDinterwordspacing

\bibitem{BsTxPower}
G.~R. MacCartney and T.~S. Rappaport, ``Millimeter-wave base station diversity
  for 5g coordinated multipoint ({CoMP}) applications,'' \emph{IEEE
  Transactions on Wireless Communications}, vol.~18, no.~7, pp. 3395--3410,
  2019.

\bibitem{NoisePower}
\BIBentryALTinterwordspacing
L.~Encyclopedia. (2021) {LTE} radio link budgeting and {RF} planning. [Online].
  Available:
  \url{https://sites.google.com/site/lteencyclopedia/lte-radio-link-budgeting-and-rf-planning}
\BIBentrySTDinterwordspacing

\end{thebibliography}

\end{document}